\documentclass[aps,reprint,groupaddress,superscriptaddress,amssymb,pre,onecolumn]{revtex4-2}
\usepackage[utf8]{inputenc}
\usepackage{color}
\usepackage{xcolor}
\usepackage{graphicx}
\usepackage{placeins}
\usepackage[margin=3cm]{geometry}

\usepackage{mdframed}
\usepackage{physics}
\usepackage{amsmath}
\usepackage{amsfonts}
\usepackage{amsthm}

\usepackage{bm}
\usepackage[symbol]{footmisc}
\usepackage[colorlinks,citecolor=blue]{hyperref}
\usepackage[capitalise]{cleveref}

\newtheorem{theorem}{Theorem}
\newtheorem{corollary}{Corollary}

\newcommand{\ie}{i.e.\ }

\newcommand{\der}[2]{\frac{\mathrm{d} #1}{\mathrm{d}#2}}
\newcommand{\lr}[1]{\langle #1 \rangle}
\newcommand{\LR}[1]{\left\langle #1 \right\rangle}
\newcommand{\ER}{Erd\H{o}s-R\'enyi}

\definecolor{deepmagenta}{RGB}{120,0,120}
\definecolor{mutedpurple}{RGB}{90,0,110}
\newcommand{\new}[1]{{#1}}

\begin{document}
\preprint{APS/123-QED}

\title{Habitat heterogeneity and dispersal network structure as drivers of metacommunity dynamics}

\author{Davide Bernardi}
\altaffiliation{Equal contribution}
\affiliation{Laboratory of Interdisciplinary Physics, Department of Physics and Astronomy ``G. Galilei", University of Padova, Padova, Italy}
\affiliation{National Biodiversity Future Center, Palermo, Italy}
\affiliation{INFN, Sezione di Padova, via Marzolo 8, 35131 Padova, Italy}
\author{Alice Doimo}
\altaffiliation{Equal contribution}
\affiliation{Laboratory of Interdisciplinary Physics, Department of Physics and Astronomy ``G. Galilei", University of Padova, Padova, Italy}
\author{Giorgio Nicoletti}
\altaffiliation{Equal contribution}
\affiliation{International Centre for Theoretical Physics, Trieste, Italy}
\author{Prajwal Padmanabha}
\altaffiliation{Equal contribution}
\affiliation{Department of Fundamental Microbiology, University of Lausanne, Switzerland}
\author{Andrea Rinaldo}
\affiliation{Department of Civil, Environmental and Architectural Engineering, University of Padova, Padova, Italy}
\author{Samir Suweis}
\affiliation{Laboratory of Interdisciplinary Physics, Department of Physics and Astronomy ``G. Galilei", University of Padova, Padova, Italy}
\affiliation{INFN, Sezione di Padova, via Marzolo 8, 35131 Padova, Italy}
\author{Sandro Azaele}
\affiliation{Laboratory of Interdisciplinary Physics, Department of Physics and Astronomy ``G. Galilei", University of Padova, Padova, Italy}
\affiliation{National Biodiversity Future Center, Palermo, Italy}
\affiliation{INFN, Sezione di Padova, via Marzolo 8, 35131 Padova, Italy}
\author{Amos Maritan}
\email{Correspondence: amos.maritan@unipd.it}
\affiliation{Laboratory of Interdisciplinary Physics, Department of Physics and Astronomy ``G. Galilei", University of Padova, Padova, Italy}
\affiliation{National Biodiversity Future Center, Palermo, Italy}
\affiliation{INFN, Sezione di Padova, via Marzolo 8, 35131 Padova, Italy}

\begin{abstract}
\noindent 
Spatial structure and species interactions jointly shape the dynamics and biodiversity of ecological systems, yet most theoretical models either neglect spatial heterogeneity or sacrifice analytical tractability. 
Here, we provide a unified microscopic, mechanistic framework for deriving effective metapopulation and metacommunity models from individual-based ecological dynamics on arbitrary dispersal networks. The resulting coarse-grained description features an effective dispersal kernel that encodes both microscopic dynamical parameters and network topology. Based on this framework, we demonstrate exact analytical results for species persistence in both homogeneous and heterogeneous landscapes, including a generalization of the classical concept of metapopulation capacity to non-uniform local extinction rates.
Incorporating stochasticity arising from finite carrying capacities, we obtain a reduced one-dimensional description that reveals universal finite-size scaling laws for extinction times and fluctuations. Extending the approach to multiple competing species, we prove that in homogeneous environments monodominance can be avoided only in a fine-tuned, marginally stable coexistence state, and that the classic metapopulation capacity gives only a necessary but not sufficient condition for persistence. We demonstrate that heterogeneous habitats can support stable coexistence, but only above a critical level of heterogeneity. Finally, we outline how additional ecological processes can be systematically incorporated within the same formalism.
Together, these results provide analytical benchmarks and a general route for constructing spatially explicit ecological theories based on an interpretable underlying mechanistic foundation.
\end{abstract}

\maketitle

\section{Introduction}
\noindent
The spatial dimension of ecological dynamics has long been recognized as important \cite{Levin1992,Tilman1994,levin2000multiple,HanOva2000, mouquet2002coexistence,chave2002,di2020dispersal, diniz2020landscape, chen2020trade, loke2023unveiling,marrec2021bitbol,abbara2023bitbol, denk2025spatial}, yet many theoretical studies continue to rely on well-mixed assumptions, where time is the only axis of change \cite{allesina2012stability,roy2020complex,altieri2021properties, barongalla2023PhysRevLett,poleygalla2023PhysRevE, hatton2024diversity}. Since the early work of Lotka and Volterra, such assumptions have been analytically convenient, but ecologically limiting, as real ecosystems are embedded in heterogeneous landscapes. Incorporating spatial structure into theory is challenging, as it significantly increases complexity. However, a quantitative understanding of the drivers and determinants of species dynamics in spatially extended systems is crucial both to uncover the mechanisms through which biodiversity arises and persists, and to anticipate the impacts of human interventions on ecosystems \cite{Diaz2006,tilman2017future,isbell2017linking}.

Multispecies dynamics on landscapes emerge from an interplay between local processes -- such as birth, death, competitive or trophic interactions -- and global processes -- such as colonization and dispersal. Developing models that can disentangle the contributions of these mechanisms, both individually and in combination, is crucial for predicting persistence and coexistence \cite{bascompte1996habitat,hanski1999,Reichenbach_2007frey,gilarranz2012spatial,baron2020dispersal,tejo2021coexistence,Wu_2022microbial_partitioning}. Without such a framework, comparisons to empirical data and the identification of vulnerable species remain elusive.

One route to incorporating space explicitly is through partial differential equation models of population densities in continuous habitats \cite{Holmes_1994_ecologyPDE,Cantrell1999PDE}. While powerful, these models are computationally demanding and analytically challenging, especially when space is not homogeneous. A more tractable intermediate framework is provided by metapopulation and metacommunity theory, in which the landscape is represented as a set of discrete habitat patches connected by dispersal \cite{hanski1999,leibold2004metacommunity}. Within patches, populations are assumed to be well-mixed, enabling models that capture both local interactions and between-patch connectivity. This abstraction has proved fruitful: in the single-species case, metapopulation models yield clear metrics for persistence, such as the metapopulation capacity introduced by Hanski and Ovaskainen, which links survival thresholds to landscape connectivity \cite{hanski1998,HanOva2000}. However, extending these approaches to multiple interacting species is more challenging. While several classical paradigms of metacommunity theory -- species sorting, patch dynamics, mass effects, and neutrality -- provide conceptual structure, empirical systems often exhibit mixtures of these processes, and generalizable measures for multispecies persistence remain limited.

Statistical physics provides a complementary perspective, aiming to describe complex systems through coarse-grained variables and emergent laws rather than microscopic detail. May’s seminal analysis of random community matrices \cite{MAY_1972} demonstrated that tools from random matrix theory could set a null expectation for stability in large ecosystems \cite{allesina2012stability,Allesina_2015_stabilityage40,Biroli2018,barongalla2023PhysRevLett}. Since then, methods drawn from spin-glass theory, disordered systems, and mean-field dynamics have been used to study ecological communities \cite{Azaele_2016,bunin2017_PRE,Galla2018,altieri2021properties,suweis2024prl,Akjouj_2024_review_GLV}. These approaches naturally emphasize universality: details of individual species and interactions can be abstracted away, allowing analysis in terms of scaling laws, phase transitions, and macroscopic stability criteria. These approaches identify broad patterns, such as phase transitions between stable and unstable community states, but they too have often relied on well-mixed assumptions that neglect spatial heterogeneity.

Metacommunity theory provides a natural meeting point between these perspectives. Single-species metapopulation models already display phase transitions between extinction and persistence, with control parameters determined by landscape connectivity \cite{HanOva2000,NicPad2023}. In systems with multiple interacting species, dispersal shapes collective dynamics and may promote nonequilibrium behaviors \cite{Reichenbach_2007frey,holland2008strong,Guichard2014,baron2020dispersal,denk20022hallatschek, nauta2024topological}. Heterogeneity in habitat quality, dispersal pathways, and competition among species resembles the quenched disorder studied in statistical physics. These parallels suggest that techniques from disordered systems, finite-size scaling, and non-equilibrium statistical mechanics are well-suited for advancing ecological theory. 

Here, we unify, generalize, and provide a more rigorous theoretical grounding for recent approaches for deriving effective metacommunity theories from the bottom up, that integrate species-specific traits and processes while explicitly accounting for landscape structure \cite{NicPad2023,PadNic2024,DoiNic2025}. This approach offers a balance between interpretability and tractability: the resulting models provide reduced descriptions that retain a clear link to the underlying ecological mechanisms. Such a connection is crucial when testing model predictions against empirical data -- such as species distributions across space -- and when tracing these patterns back to landscape features and species traits, ideally disentangling their respective contributions.

The article is organized as follows. In \cref{sec:micromodel}, we present the core ideas in a simplified scenario involving a single species and basic ecological processes: space-limited local demography and network-based dispersal and colonization \cite{NicPad2023}. We then show analytical results for single-species persistence in arbitrarily structured landscapes, extending the classical concept of metapopulation capacity to heterogeneous habitats where local extinction rates may vary (\cref{sec:onespecies-theorems}). Next, in \cref{sec:stochastic}, we study the effect of stochasticity \cite{DoiNic2025}. We then introduce interspecific interactions into the model in \cref{sec:multispec-model}, beginning with competition \cite{PadNic2024}. In \cref{sec:multispec-results}, we derive further analytical results on the coexistence of species in both homogeneous and heterogeneous habitats. Finally, we show in \cref{sec:outlook} how additional ecological processes can be systematically included, paving the way for metacommunity models that strike a balance between reduced complexity and a mechanistic link to the underlying ecological processes. Proofs of theorems in the main text are presented in the appendices.

\section{Metapopulation models from microscopic dynamics}\label{sec:micromodel}
\noindent We introduce a microscopic description of a metapopulation model on arbitrary networks \cite{NicPad2023}. We first consider the dynamics of a single species in a dispersal network of $N$ habitat patches, each with a finite number $M_i$ of colonizable sites. The network is described by a weighted adjacency matrix $\hat{A}$, whose elements $A_{ij}$ denote the strength of the dispersal from the $i$-th habitat to the $j$-th one, with $A_{ii} = 0$. In this setting, we describe all ecological processes through individual-based transitions \cite{gardiner,McKaneNewman2004_PRE}. We assume that each patch is inhabited by a settled population that does not move. If we denote a settled individual in patch $i$ by $S_i$, we assume that each individual can die with a patch-dependent death rate $e_i>0$. 
\begin{equation}
        S_i \xrightarrow{e_i} \varnothing_i,
    \label{eqn:reactions_settled_singlespecies}
\end{equation}
where $\varnothing_i$ denotes an empty site in the same patch.

Although settled individuals do not move, they can produce new individuals of the same species that explore the network and attempt to colonize other patches. We denote one such explorer in the $i$-th patch with $X_i$. In some ecological contexts, this distinction between settled populations and explorers has a direct biological interpretation -- for example, plants and their seeds, or sessile adults and dispersing larvae. In other contexts, it serves as a conceptual distinction reflecting differences in the life histories or time scales of ecological processes, which are fundamental for deriving general metapopulation models from microscopic dynamics, as we will see. Explorers' birth is governed by the reaction
\begin{equation}
        S_i \xrightarrow{c_i \, C_{ij}} S_i + X_j,
    \label{eqn:reaction_exbirth_singlespecies}
\end{equation}
where $c_i$ acts as a birth rate for the explorers from patch $i$, and $C_{ij}$ denotes the feasibility of exploration from patch $i$ to patch $j$. Similarly, explorers can move between two patches with a rate $D_{ij} = D \, W_{ij}$, where $D$ is a baseline exploration rate and $W_{ij} \ne 0$ if and only if explorers can move directly from the $i$-th habitat to the $j$-th  habitat patch. We will always make the assumption that $W_{ij} = A_{ij}$, i.e., that explorers can only move one step at a time along the dispersal network edges. However, our results can be applied to other scenarios, such as one where multiple steps are taken at once. Then, explorers die with a rate $\gamma$ or attempt to colonize the patch they are in with a rate $\lambda$. We write these reactions as
\begin{equation}
    \begin{gathered}
        X_i \xrightarrow{D \, W_{ij}} X_j \\
        X_i \xrightarrow{\gamma} \Phi \\
        \quad X_i + \varnothing_i \xrightarrow{\lambda / M_i} S_i , \quad X_i + S_i \xrightarrow{\lambda / M_i} S_i
\label{eqn:reactions_explorers_singlespecies}
    \end{gathered}
\end{equation}
where $\Phi$ denotes an empty unbounded space for explorers. The rate $\lambda/M_i$ appearing in \cref{eqn:reactions_explorers_singlespecies} takes into account both the rate at which explorers attempt colonization, $\lambda$, and the fact that, when they do, they can choose one of the $M_i$ sites of patch $i$. \new{A point of note here is our modeling choice. New individuals of a species can only be created by a combination of both diffusion of explorers and their subsequent colonization. This is conceptually different from diffusion of settled individuals: diffusion redistributes existing individuals on a graph and does not imply the creation of new individuals. Below, we show that separating these two processes leads to emergent coarse-grained models with interpretable dispersal kernels, which would otherwise be absent.}

We denote the number of individuals with $[\cdot]$ so that the state of the model is fully specified by $([\vec{S}], [\vec{X}]) = ([S_1], \dots, [S_N], [X_1], \dots, [X_N])$ since, at all times, we have that $[S_i] + [\varnothing_i] = M_i$ for all $i$. In particular, we are interested in the density of individuals in each patch $i$, $p_i = [S_i]/M_i \in [0,1]$, and we first consider the limit in which the number of colonizable sites $M_i$ is large, so we can ignore stochastic effects. We will relax this assumption later. We obtain the equations 
\begin{gather}
\label{eqn:rate_equations_singlespecies}
        \dot{p}_i = -e_i p_i + \lambda(1 - p_i) x_i \\
        \dot{x}_i = - (\lambda + \gamma) x_i + \sum_{j = 1}^N p_j c_j\frac{M_j}{M_i} C_{ji} + \frac{D}{M_i} \sum_{j = 1}^N\left(x_j M_jW_{ji} - x_i M_i W_{ij}\right) \nonumber
\end{gather}
where we also introduced for convenience the rescaled number of explorers $x_i = [X_i] / M_i$. \Cref{eqn:rate_equations_singlespecies} also corresponds to the evolution of the averages
\begin{gather*}
    \ev{S_i}_p(t) = \sum_{[S_i] = 0}^{M_i} [S_i] \, p([S_i], t) \\
    \ev{X_i}_p(t) = \sum_{[X_i] = 0}^{\infty} [X_i] \, p([X_i], t)
\end{gather*}
and can be seen as the first order of the Kramers-Moyal expansion of the master equation for the probability $p(\vec{[S]}, \vec{[X]}, t)$, where $\vec{[S]} = ([S]_1, \dots, [S]_N)$ with the set of reactions in \cref{eqn:reactions_settled_singlespecies,eqn:reaction_exbirth_singlespecies,eqn:reactions_explorers_singlespecies}. 

\subsection{Fast-exploration regime}\label{sec:fast}
\noindent Although \cref{eqn:rate_equations_singlespecies} fully describes the deterministic limit of our microscopic model, we now assume that only settled individuals are visible and seek to derive an effective metapopulation dynamics involving only the settled population density. Specifically, we take the limit in which explorers are much faster than the local dynamics on patches \cite{NicPad2023} (see Figure \ref{fig:single_species_results}a). \new{By grouping together the diagonal decay and the diffusion term in the explorers' equation, the quasi-stationary condition $\dot{x}_i = 0$ can be rewritten as
\begin{equation}
     \sum_{j = 1}^N \left[(\lambda + \gamma) \delta_{ij} - D \left(\frac{M_j}{M_i} W_{ji} - \delta_{ij} \sum_{k = 1}^N W_{ik}\right)\right]x_j  = \sum_{j = 1}^N p_j c_j\frac{M_j}{M_i} C_{ji} \nonumber
\end{equation}
where, after collecting a factor $\lambda$, the l.h.s. depends on the matrix
\begin{equation}
\label{eqn:Fmatrix_singlespecies}
    F_{ij} = \left(1 + \frac{\gamma}{\lambda}\right)\delta_{ij} + \frac{D}{\lambda} \left(\delta_{ij}\sum_{k=1}^N W_{ik} - \frac{M_j}{M_i}W_{ji}\right).
\end{equation}
}\cref{eqn:Fmatrix_singlespecies} only depends on two dimensionless parameters,
\begin{equation}
    \label{eqn:adim_params_singlespecies}
    f = \frac{D}{\lambda}, \qquad g = \frac{\gamma}{\lambda},
\end{equation}
with $f$ quantifying the exploration efficiency in terms of how many patches are visited before a colonization attempt and $g$ the explorers' lifespan. Clearly, if $f \ll 1$, \ie $D \ll \lambda$, explorers will remain close to the originating patch. On the other hand, if $D \gg \lambda$ they will explore far-away patches before settling on one. Similarly, $g \ll 1$ denotes the limit in which explorers will survive long enough to always attempt colonization before dying, whereas $g \gg 1$ makes exploration harder.

This enables us to write the relation
\begin{equation}
\label{eqn:lambdaxi_singlespecies}
    \lambda x_i = \sum_{j = 1}^N\sum_{k = 1}^N p_j c_j \frac{M_j}{M_i} C_{jk} F^{-1}_{ik}
\end{equation}
where $F^{-1}_{ij}$ is a shorthand notation for the elements of the inverse of $\hat{F}$. By inserting \cref{eqn:lambdaxi_singlespecies} into the rate equation for the settled population, \cref{eqn:rate_equations_singlespecies}, we end up with the effective evolution
\begin{equation}
    \label{eqn:settled_dynamics_singlespecies}
    \dot{p}_i = -e_i p_i + (1-p_i) \sum_{j = 1}^N  c_j K_{ij} p_j
\end{equation}
where we introduced the kernel matrix
\begin{equation}
\label{eqn:matrix_kernel_singlespecies}
    K_{ij} = \frac{M_j}{M_i} \sum_{k = 1}^N C_{jk} F^{-1}_{ik}
\end{equation}
which describes the effective colonization rate from patch $j$ to patch $i$ due to the effect of the explorers. As expected, large patches have a larger influence on smaller ones through the pre-factor $M_j / M_i$. Importantly, the effective dynamics, through the kernel $\hat K$ in \cref{eqn:matrix_kernel_singlespecies}, summarizes the dependence of the detailed reactions and the network topology, showing how the microscopic features of the model are instrumental in determining the properties of the metapopulation.

\subsection{Kernel-based encoding of structure and dynamics}
\noindent In order to understand the properties of the kernel, we need explicitly to invert the matrix $\hat{F}$ in \cref{eqn:Fmatrix_singlespecies}. To do so, we assume that the matrix  
\begin{equation}
\label{eqn:general_laplacian_singlespecies}
    L_{ij} = \delta_{ij}\sum_{k=1}^N W_{ik} - \new{\frac{M_j}{M_i}}W_{ji}
\end{equation}
is diagonalizable, \ie $\hat{L} = \hat{V} \hat{\Omega} \hat{V}^{-1}$, where $\hat \Omega$ is a diagonal matrix. We note that if $W_{ij} = A_{ij}$ (the adjacency matrix introduced above) \new{and patches are identical ($M_i = M_j$)}, then \cref{eqn:general_laplacian_singlespecies} corresponds exactly to the standard combinatorial Laplacian, with $\sum_k A_{ik} := q_i$ as the out-degree of patch $i$. Other relevant choices are possible, such as $W_{ik}= A_{ik}/q_i$, which corresponds to the random walk Laplacian.

In general, we will focus on the scenario where explorers move according to the dispersal network described by $\hat{A}$. Then, we can immediately invert $\hat{F}$, which is diagonal in the eigenspace of $\hat{L}$, to find the explicit form of the kernel:
\begin{equation}
\label{eqn:matrix_kernel_explicit_singlespecies}
    K_{ij} = \frac{M_j}{M_i}\sum_{k = 1}^N C_{jk} \sum_{l = 1}^N \frac{V_{kl}(V^{-1})_{li}}{1 + g + f \omega_l}
\end{equation}
where $\omega_k$ is the $k$-th eigenvalue of $\hat{L}$.
In suitable cases, \cref{eqn:matrix_kernel_explicit_singlespecies} can be interpreted as a weighted sum over all possible paths between two patches \cite{NicPad2023}. Furthermore, it manifestly incorporates both dynamical features - through $f$ and $g$ - and topological ones - through $\hat{C}$ and the eigenspace of $\hat{L}$.

Several choices for the feasibility of exploration $C_{ij}$ are possible as well, which should depend, in general,  on the efficiency of exploration $f$. In particular, if $f \to \infty$, we assume that there exists a maximal explorability $\xi$, which corresponds to the maximal dispersal capacity of the species. Conversely, if $f \to 0$, exploration should not be possible, and $C_{ij}$ should vanish. Hence, we expect that
\begin{equation}
    C_{ij}(f)\xrightarrow{f \ll 1} 0, \qquad C_{ij}(f) \xrightarrow{f \gg 1} \xi.
\end{equation}
The simplest choice would be to set $C_{ij} \propto \delta_{ij}$, so that explorers are created in the same patch. In order to avoid self-colonization - which is particularly relevant if $f \ll 1$ - another option is to take $C_{ij} \propto A_{ij}$, so that explorers are created in neighboring patches. This choice is equivalent to assuming that explorers need to take at least one exploration step before attempting colonization. Then, a simple parametric form is
\begin{equation}
    C_{ij} = \frac{\xi}{1 + 1/f} A_{ij} \; ,
\end{equation}
but our approach naturally allows for other choices as well, leading to a vast array of possible phenomenologies. 

In  \Cref{fig:single_species_results}b, we show how the kernel changes with the structure of the dispersal network. We consider three exemplary topologies. In a ring network, the kernel decays exponentially with the network distance $d_{ij}$ between patch $i$ and patch $j$. The characteristic decay length increases with $f$, suggesting a direct relation between exploration efficiency and the average dispersal distance of the species. However, in more complex topologies, the kernel elements are not solely a function of the network distance. \new{Even though there is a shortest distance between two nodes (denoted by $d_{ij}$), there are many paths that can be taken by the explorers. Despite this, we observed in generic topologies that, on average, the kernel values decay exponentially with the shortest distance between two nodes. The characteristic decay rate of this exponential correlates with the value of $f$, as noted above for the ring topology. The variance of kernel values at a given distance} becomes especially relevant at large $f$. In particular, in \Cref{fig:single_species_results}b, we also consider a small-world dispersal network, introducing long-range connections between patches that may be otherwise very far apart, and a Barabasi-Albert topology, where a few patches behave as hubs and connect to a large number of other nodes. 

\begin{figure}
    \centering
    \includegraphics[width=1\textwidth]{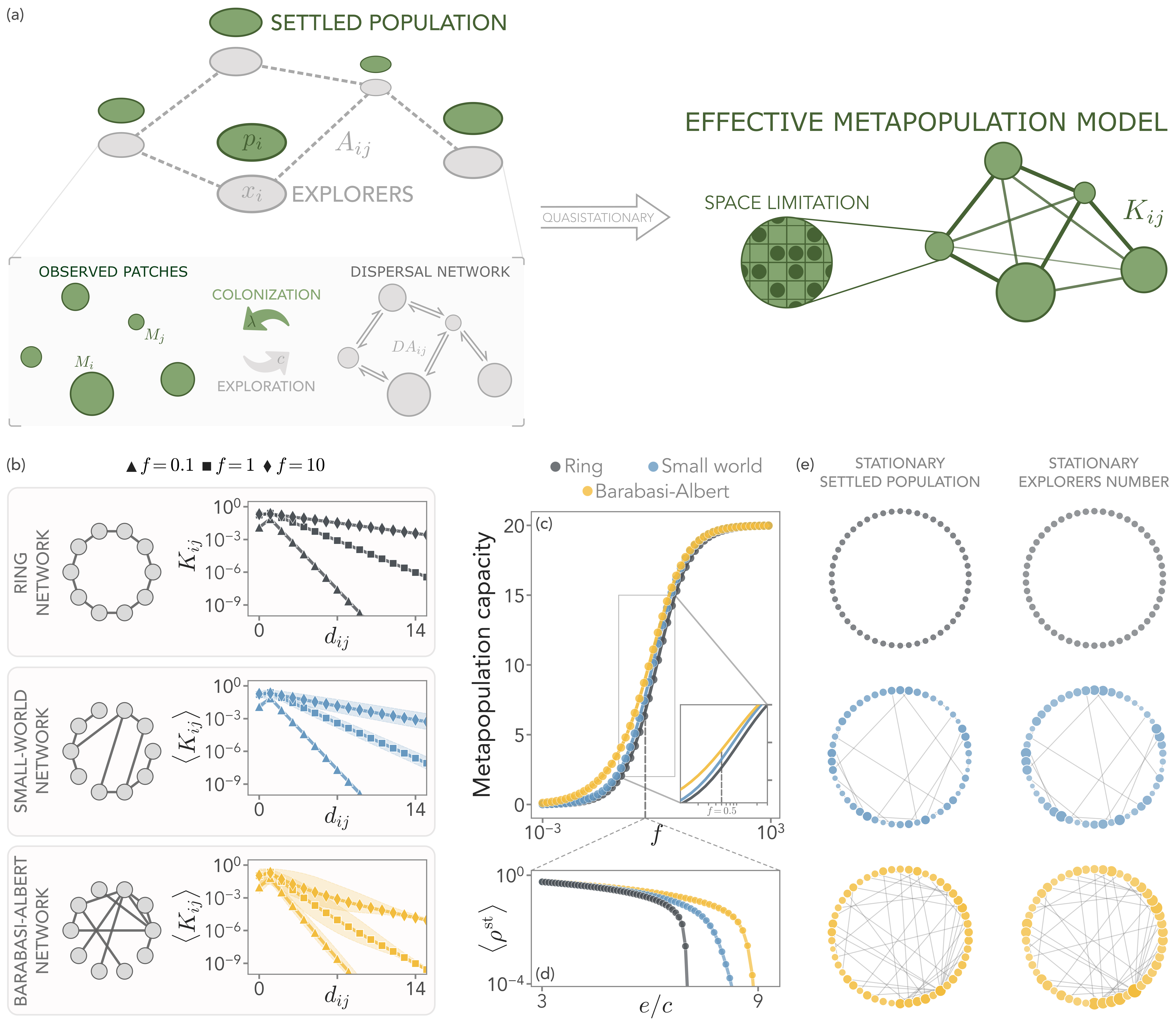}
    \caption{\new{(a) Sketch of the model. We consider the dynamics of a settled population (green, $p_i$), which resides in habitat patches, and of explorers (gray, $x_i$) that move along a dispersal network that connect the habitats, described by an adjacency matrix $A_{ij}$. Each habitat patch has a different number of colonizable sites ($M_i$). Explorers are born with a patch-dependent rate $c_i$, move between neighboring patches with a rate $D A_{ij}$, and attempt to colonize empty sites with a rate $\lambda$. In the fast-exploration regime, this leads to an effective metapopulation model with an explicit colonization kernel $K_{ij}$.} (b) The colonization kernel depends manifestly on the topology of the dispersal network. In general, it is not a function of the network distance $d_{ij}$ alone, but rather depends on all paths connecting two habitat patches. (c-e) The metapopulation capacity, which determines a species' survival ability, depends on both the exploration efficiency $f$ and the network topology. In particular, more heterogeneous dispersal networks are more advantageous at intermediate and low $f$, affecting the stationary population and explorers (panels (d) and (e)). Figure adapted from \cite{NicPad2023}.}
    \label{fig:single_species_results}
\end{figure}

\subsection{Metapopulation persistence}
\label{sec:metapop}
\noindent In full generality, \cref{eqn:settled_dynamics_singlespecies} admits a trivial solution where the focal species is extinct, \ie $p_i = 0$ for all patches. We can characterize the linear stability of this solution by computing the leading eigenvalue of the Jacobian matrix
\begin{equation*}
    J_{ij}(\vec{p} = 0) = \frac{\partial \dot{p}_i}{\partial p_j}\biggr|_{\vec{p} = 0} = -e_i \delta_{ij} + c_j K_{ij}.
\end{equation*}
In particular, we can first focus on the case considered by previous works in which all patches are equivalent, so that $M_i = M$, $e_i = e$, and $c_i = c$. Then, we find that the extinct state is unstable if and only if
\begin{equation}
    \label{eqn:survival_condition_singlespecies}
    \lambda_M > \frac{e}{c}
\end{equation}
where $\lambda_M$ is the maximum eigenvalue of the kernel $\hat{K}$, and it is referred to as the \emph{metapopulation capacity} of the landscape \cite{HanOva2000}. \new{Intuitively, this measure means that the ratio between extinction and colonization rates defines a threshold for the largest eigenvalue of the dispersal kernel: if the eigenvalue exceeds this threshold, the species is able to persist on the network. Conversely, if the network yields an eigenvalue below the threshold, the species will eventually go extinct in all patches. Therefore, a higher metapopulation capacity enables the persistence of species with a less favorable (i.e., larger) extinction-to-colonization ratio, highlighting the role of the environment in shaping species persistence.}

In \Cref{fig:single_species_results}c, we show that the structure of the dispersal network directly affects the metapopulation capacity. The survival of a species tends to increase with the heterogeneity of the network, especially at intermediate and low $f$, where the topological features may help a species colonize and survive.

\section{Exact results on the metapopulation capacity of a focal species}
\label{sec:onespecies-theorems}
\noindent We now prove several exact results on the focal species persistence based on \cref{eqn:settled_dynamics_singlespecies}. The proofs of all theorems and corollaries are in \cref{app:proofs-one}. In the following, we can set $c_i = 1$ for all $i$, which is equivalent to a rescaling of the colonization kernel, and does not restrict the generality of the results. 

\subsection{Generalized metapopulation capacity in heterogeneous habitats with non-uniform local extinction rates}
\noindent We first generalize the classic result on the metapopulation capacity of a landscape from \cref{sec:metapop}, which holds only for the case of a patch-independent extinction rate. To this end, we leverage the fact that the entries of the colonization kernel are all positive if the graph has a single connected component. However, we can adopt the even more general assumption that $\hat{K}$ is irreducible.
Given the assumption stated above,
\begin{theorem}
Given the following metapopulation dynamics
\begin{equation}
	\dot p_i(t) = -e_ip_i(t) + \Big(1-p_i(t)\Big)\sum_{j=1}^N K_{ij} c_j p_j(t),
	\label{eq:model-one-spec-lean}
\end{equation}
where $e_i>0$ and the kernel matrix $\hat{K}$ is irreducible. Let $\hat{\mathcal{E}}$ be the matrix with elements $\mathcal{E}_{ij} = e_i \delta_{ij}$, and $\hat{\mathcal{C}}$ be the matrix with elements $\mathcal{C}_{ij} = c_i \delta_{ij}$. Then, the generalized landscape matrix
\begin{equation}
    \hat{\mathcal{M}} = \hat{K}\hat{\mathcal{C}}\hat{\mathcal{E}}^{-1}
\end{equation}
is also irreducible, and has a Perron-Frobenius positive eigenvalue $\lambda_\mathrm{max} > 0$. Then, for all non-trivial initial conditions, $\vec p(0)\neq 0$, and $p_i(0)\in [0,1]$, $\lim_{t \to \infty}\vec p(t)= 0$, i.e., $\vec p=0$ is globally stable if $\lambda_\mathrm{max} < 1$ and unstable if $\lambda_\mathrm{max} > 1$.
\label{th:theorem1}
\end{theorem}

\noindent
This theorem leads to the following statement:
\begin{corollary}
When $\lambda_\mathrm{max}>1$, the stationary state of \cref{eq:model-one-spec-lean} is non-trivial. Specifically, if a non-trivial stable stationary state exists, then the population is strictly positive in all patches, \ie $p_i>0,\, \forall i$.
\label{th:firstcorollary}
\end{corollary}
\noindent
We now observe that the classical result on the metapopulation capacity of Hanski and Ovaskainen, assuming an exponentially decaying dispersal kernel \cite{HanOva2000}, can be derived as a corollary of \cref{th:theorem1}:
\begin{corollary}
Let
\begin{equation}
e_i = \frac{e}{A_i}, \quad c_i = c\,A_i, \quad \mathrm{and} \quad K_{ij} = e^{-\alpha d_{ij}},
\end{equation}
where $A_i$ denotes the area of patch $i$ and $d_{ij}$ the distance between patches $i$ and $j$. Then, the long-term behavior of \cref{eq:model-one-spec-lean} is non-trivial if $\delta = e / c < \lambda_M$, where $\lambda_M$ is the largest eigenvalue of the symmetric landscape matrix $\mathcal{M}^\mathrm{H}_{ij} \equiv A_i e^{-\alpha d_{ij}} A_j$.
\label{th:HO-assumptions}
\end{corollary}
\noindent
While the rigorous proof of both corollaries is provided in \cref{app:proofs-one}, we immediately have that \cref{eqn:survival_condition_singlespecies} can be rewritten as $\lambda_M \, c \, e^{-1} > 1$, in agreement with \cref{th:theorem1}. A noteworthy implication of \cref{th:theorem1} is that the local extinction and colonization rates of every landscape patch contribute to the leading eigenvalue and thus affect the species persistence threshold. Consequently, this threshold is a global property of the system rather than being determined by a few particularly favorable patches. Thus, we define the generalized metapopulation capacity, $\lambda_\mathrm{max}$, as the maximum eigenvalue of the matrix $\hat{K}\hat{\mathcal{C}}\hat{\mathcal{E}}^{-1}$.

\subsection{Stability and uniqueness of steady-state solutions in homogeneous habitats}
\noindent {We now consider the special case of systems that are spatially homogeneous, a scenario in which further exact results can be proven. Note that in such a setting, $e_i = e \; \forall i$. Then, the stability of both trivial (global extinction) and non-trivial steady states can be studied explicitly. Furthermore, we prove that the non-trivial steady state is unique. Strict spatial homogeneity implies dynamics taking place in a finite $d$-dimensional lattice $\Lambda = \{1,2,\dots,N-1\}^d$ with spacing set equal to $1$, periodic boundary conditions, and $K_{ij}$ can depend only on $j-i \mod N,\, \forall i,j\in \Lambda$. However, we are able to prove the following theorem under the less restrictive hypothesis that $\hat{K}$ is irreducible and that row sums are independent of the row index.
\begin{theorem}
	Consider the metapopulation dynamics:
	\begin{equation}
		\dot p_i= -e p_i + \big(1-p_i\big)\sum_{j} K_{ij}p_j.
		\label{eq:model-one-spec-const-row-sum}
	\end{equation}
Assume that $\hat{K}$ is irreducible and that its row sums are constant $\tilde{  K}_0 = \sum_j K_{ij}$.  Then, 
\begin{itemize}
\item When $e>\tilde{  K}_0$, the global extinction state is globally stable; a non-trivial steady state is not feasible.
\item When $e<\tilde{ K}_0$, the global extinction state is unstable; the non-trivial, spatially-uniform stationary solution is unique and locally stable.
\end{itemize}
\label{th:theorem-TI-singlespecies}
\end{theorem}
\noindent
In \cref{app:proofs-one} we provide both a proof for the general case of a constant-row matrix, as in the theorem's statement, and an alternate proof valid only in the translation-invariant case, which uses explicit Fourier modes. Previous works showed a stronger result, the global stability of the non-trivial state \cite{Lajmanovich_1976,Fall_2007}. The proofs of \cref{th:theorem-TI-singlespecies} reported in \cref{app:proofs-one} are based on simpler, self-contained arguments.}

\section{Stochastic Metapopulation Dynamics and Finite-Size Effects}\label{sec:stochastic}
\noindent
Having established the basic framework and illustrated it with both numerical and exact results, we now examine the effects of a finite patch size $M_i$, which acts as a finite carrying capacity and introduces fluctuations in the settled population. Rather than simply adding demographic noise to \cref{eqn:settled_dynamics_singlespecies}, a consistent incorporation of stochastic effects within our framework requires returning to the underlying microscopic model and moving beyond the first order Kramers-Moyal expansion we performed before.

\subsection{From deterministic dynamics to stochastic fluctuations}
\noindent
We begin with the microscopic model introduced in \cref{sec:micromodel}, but we make some simplifications for analytical tractability. We assume there is no inherent explorer death ($\gamma = 0$),  that all patches have the same carrying capacity ($M_i = M$), that explorer creation reflects the network ($C_{ij} \propto A_{ij}$), and without loss of generality, $c=1$. We set $p_i = [S_i]/M \in [0,1]$ and $x_i = [X_i]/M$. To analyze the stochastic dynamics, we consider the corresponding master equation and perform a Kramers-Moyal expansion, using the inverse carrying capacity $1/M$ as the expansion parameter. Truncating the expansion at first order recovers the deterministic model discussed above, while extending it to second order yields the following Fokker–Planck equation \cite{DoiNic2025}:
\begin{equation}
\partial_t \mathcal{P}(\vec{z}, t) = 
- \sum_i \partial_i \left[ \mathcal{A}_i(\vec{z}) \mathcal{P}(\vec{z}, t) \right]
+ \frac{1}{2 M} \sum_{i,j} \partial_i \partial_j \left[\mathcal{D}_{ij}(\vec{z})  \mathcal{P}(\vec{z}, t) \right],
\label{eq:FP}
\end{equation}
where $\vec{z} = (\vec{p}, \vec{x})$ is the $2N$-dimensional state vector, $\vec{A}(\vec{z})$ is the drift vector 
\\
\begin{align}
	& \mathcal{A}_i\equiv \mathcal{A}^{p}_i = 
	\lambda \, x_i (1-p_i) -  e_i p_i, \quad i\leq n
	\nonumber  \\
	& \mathcal{A}_i \equiv \mathcal{A}^{x}_i = 
	D \sum_j 
	\left(
	W_{ji} x_j -W_{ij} x_i  
	+ C_{ji} p_j \right) - \lambda \,x_i, \quad  i> n,
	\label{eq:FP-drift}
\end{align}
where we assume that $W_{ij}=A_{ij}$ and $ \hat{\mathcal{D}}$ is a $2N \times 2N$ diffusion matrix with the block structure
\begin{align}
	&\hat{\mathcal{D}}(\vec{z}) = 
	\begin{pmatrix}
		\hat{\mathcal{D}}^{pp} 
		& \hat{\mathcal{D}}^{p x} \\
		\hat{\mathcal{D}}^{ x p} 
		&  \hat{\mathcal{D}}^{xx}
	\end{pmatrix},
	\label{eq:FP-diffusion}
\end{align}
where the $N\times N$ blocks are
\begin{align}
	\mathcal{D}^{pp}_{ij} = & \left[ e_i p_i + \lambda  x_i (1 - p_i) \right]\delta _{ij} 
	\nonumber \\
	\mathcal{D}^{p x}_{ij} = & \mathcal{D}^{x p}_{ij} = - \lambda x_i (1-p _i)  \delta _{ij} 
	\nonumber \\
	\mathcal{D}^{xx}_{ij} = & \left[\sum_k (p _k C_{ki}+x_i W_{ik}+x_k W_{ki})+\lambda  x_i\right] \delta_{ij}  	
	 -(x_i W_{ij}+x_j W_{ji}) (1- \delta _{ij}) \; .
\end{align}
As expected, the diffusion term of \eqref{eq:FP} vanishes in the limit of large local capacity $M \to \infty$.
Although the Fokker–Planck equation represents a significant dimensional simplification compared to the full master equation, its form still poses several challenges, as the drift and diffusion terms, \cref{eq:FP-drift,eq:FP-diffusion}, are complex and coupled both within and across patches. To simplify, we consider a fully connected dispersal network, \ie, an adjacency matrix $A_{ij} \propto 1/N (1-\delta_{ij})$, representing homogeneous habitat patches. Then, we perform a time-scale separation and write a one-dimensional stochastic differential equation for the effective density of settled individuals $p$:
\begin{equation}\label{eq:stochSDE}
\dot{p} = \mathcal{\tilde{A}}(p) + \sqrt{\frac{1}{M} \, \tilde{\mathcal{D}}(p)} \, \xi(t)\,  
\end{equation}
where $\xi(t)$ denotes a Gaussian white noise, while $\mathcal{A}(p)$ and $\tilde{\mathcal{D}}(p)$ are the effective drift and diffusion coefficients, respectively (see \cref{app:stochmetapop} for details on their derivation and explicit forms).
\Cref{eq:stochSDE} is defined on the domain $0 < p \leq 1$, with $p = 0$ as an absorbing boundary corresponding to global extinction and $p = 1$ as a reflecting boundary representing the maximum local capacity.
In the deterministic limit, \cref{eq:stochSDE} recovers the metapopulation capacity $\lambda_M$, which separates extinction from persistence, in agreement with \cref{sec:metapop} and \cref{th:theorem1}. For finite $M$, extinction is inevitable as $t \to \infty$, and $p_{\mathrm{st}} = 0$ is the only true stationary state. Nonetheless, the condition $e < \lambda_M$ still defines a survival regime corresponding to a long-lived metastable state, where extinction occurs only through large fluctuations.
Simulations of the full microscopic model show excellent agreement with the one-dimensional reduction \cref{eq:stochSDE} in the survival regime, and only minor deviations in the mean-driven extinction regime that do not affect the qualitative behavior (\Cref{fig:stochmetapop}a-b).

\subsection{Finite-size scaling of the extinction-time distribution}
\noindent
Since the noise amplitude scales as $1/M$, systems with smaller carrying capacities experience stronger intrinsic fluctuations and a higher probability of noise-induced extinction.
We study the survival probability $S(t|M)$ at a finite system size $M$, focusing on the deterministic transition point ($e/c = \lambda_M$), where we expect a possible critical behavior. Indeed, numerical simulations of \cref{eq:stochSDE} reveal that the survival probability collapses as:
\begin{equation}
S(t|M) \sim f\left( \, t \, M^{\phi} \right),
\label{eq:StScaling}
\end{equation}
shown in \Cref{fig:stochmetapop}c-d. This behavior indicates that the critical regime exhibits classical finite-size scaling with a universal exponent $\phi = -1/2$.
The scaling behavior of the survival probability and of the first two moments of the extinction time distribution can be investigated analytically by studying the scaling behavior of the backward Fokker-Planck operator (see \cref{app:stochmetapop}). This analysis yields
\begin{align}
    \label{eq:scaling}
   S  = t^{-\alpha} \, \hat{S}\left( t M^{\phi},p_0 M^{\gamma}, \Delta M^{\eta} \right)\,,
      \quad  
    \langle T \rangle =  M^{\beta} \, \hat{T}_1 (p_0 M^{\gamma}, \Delta M^{\eta})\,,
        \quad
    \langle T^2\rangle= M^{\delta} \,\hat{T}_2 (p_0 M^{\gamma}, \Delta M^{\eta})
\end{align}
where $\Delta = (\lambda_M - e/c)/\lambda_M$ quantifies the distance from the deterministic transition threshold, and $p_0$ is the initial condition. The exponents are connected through the scaling relations $\beta=-\phi(1-\alpha),  \,  \delta=-\phi(2-\alpha)\;$. Their numerical values are reported in \cref{tab:scaling} (\cref{app:stochmetapop}). Importantly, these exponents characterize a specific universality class of noise-driven extinction dynamics. Modifications to the underlying individual reactions may change the functional form of the SDE and alter the critical exponents. Therefore, analyzing such scaling behaviors can provide a means to identify different underlying ecological mechanisms, using the present framework as a template.

\begin{figure*}[t!]
    \centering
    \includegraphics[width=\textwidth]{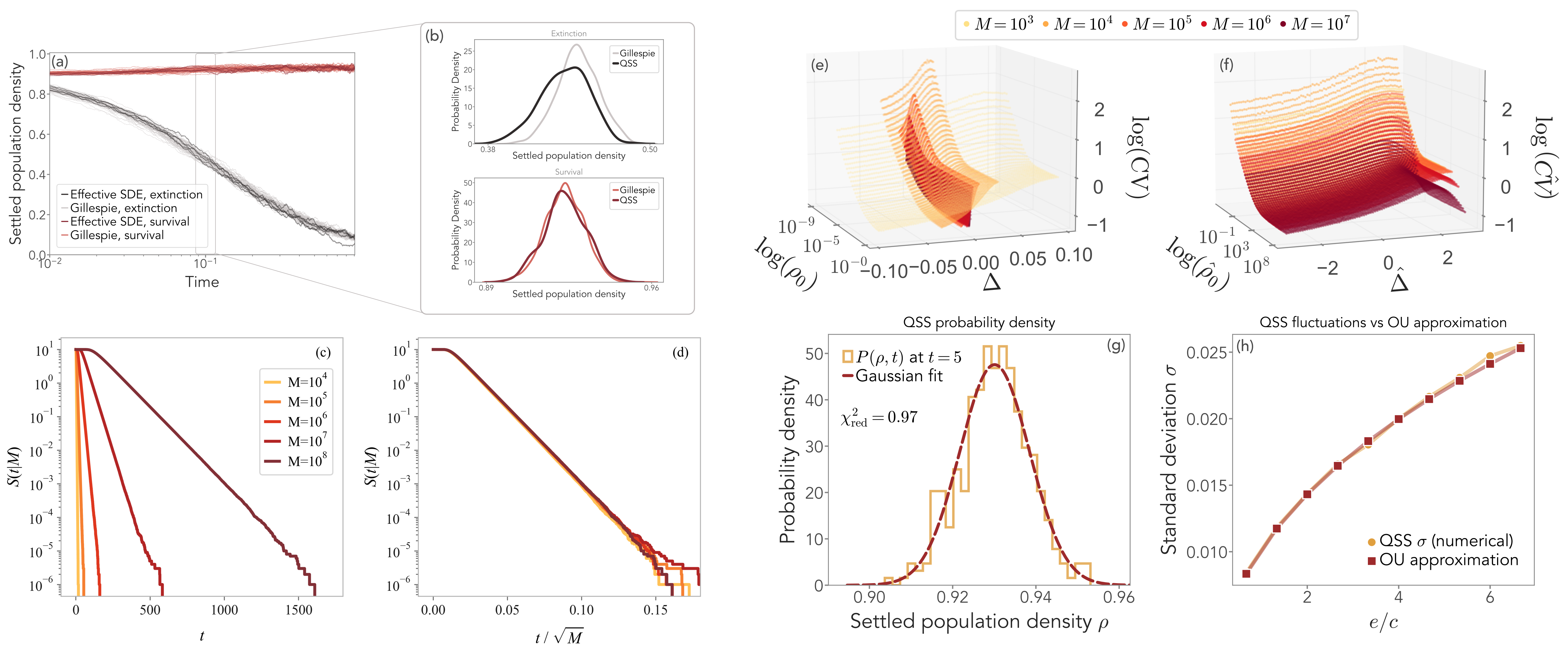}
    \caption{
    (a–b) Comparison of stochastic trajectories from the effective quasi-stationary SDE \cref{eq:stochSDE} and the full Gillespie simulation in the survival and extinction regimes. 
    (c–d) Finite-size scaling of the survival probability at criticality: (c) raw curves of survival probability as functions of time for different carrying capacities $M$, (d) collapse under the scaling ansatz \cref{eq:StScaling}.
    The initial condition is fixed at $p_0=1/2$. 
    Averages are obtained over $10^7$ numerical realizations of the dynamics in \cref{eq:stochSDE}.
    (e–f) Coefficient of variation of extinction times as a function of initial condition $p_0$ and deviation from criticality $\Delta$ for different values of the local carrying capacity $M$: (e) unscaled, (f) rescaled using finite-size scaling of the first two moments \cref{eq:scaling}. Averages are obtained over $10^7$ numerical realizations of the dynamics in \cref{eq:stochSDE}.
   (g–h) Fluctuations in the metastable regime: (g) distribution of settled densities obtained from 350 independent realizations with Gaussian fit; (h) standard deviation $\sigma$ vs. $e/c$, compared with the analytical estimate based on the Ornstein--Uhlenbeck (OU) approximation. The numerical values of $\sigma$ are obtained from Gaussian fits to distributions built from 250 independent realizations. Adapted from ref. \cite{DoiNic2025}.
    }
    \label{fig:stochmetapop}
\end{figure*}

\subsection{Linear approximation in the metastable regime}
\noindent
Building on the effective one-dimensional dynamics derived in the previous section, we show that key features of stochastic extinction can be described using simplified observables that retain ecological interpretability. In particular, we analyze the coefficient of variation (CV) of extinction times, defined as the ratio between the standard deviation and the mean, \( \mathrm{CV}(T) = \sqrt{\langle T^2 \rangle - \langle T \rangle^2} / \langle T \rangle \).
This dimensionless measure of relative variability exhibits a finite-size scaling collapse, as expected, and develops a pronounced maximum at a finite rescaled distance from criticality \( \hat{\Delta} = \Delta  \sqrt{M} > 0 \). This peak, located within the metastable survival regime, signals a fluctuation-driven transition associated with maximal uncertainty in extinction time (\Cref{fig:stochmetapop}e-f).
In this regime, extinctions are caused by a large fluctuation from the metastable state, which can be captured by a linear approximation of the effective dynamics \cref{eq:stochSDE} around the deterministic fixed point. This effective Ornstein-Uhlenbeck process (\cref{app:stochmetapop}) approximation closely agrees with numerical simulations (\Cref{fig:stochmetapop}g-h), demonstrating how a coarse representation of the stochastic system can be used effectively to link measurable population-level statistics to underlying ecological mechanisms.
\\\\ \indent
\new{
\subsection{Empirical outlook}
\noindent
The scaling predictions derived here also suggest possible routes for empirical tests. The most direct observables are extinction-time distributions, survival probabilities, and fluctuations of patch occupancy around a metastable state. These quantities become especially informative when estimates of effective local carrying capacity and landscape connectivity are available, since the deterministic metapopulation capacity $\lambda_M$ alone does not constrain extinction timescales. From the empirical side, the key requirement would be to compare systems or habitat patches that differ in an effective local capacity, such as patch area, resource availability, or typical census size, while retaining comparable dispersal conditions. Examples of suitable candidates include the Baltic rock-pool \textit{Daphnia} dataset (17 years, $>$500 pools of varying size) \cite{Pajunen2003}, which provides long-term extinction records as a function of pool volume as a proxy for $M$, or island and lake metapopulations \cite{HanOva2000}, where patch-level occupancy and richness data can be used to probe occupancy fluctuations near criticality.  Particularly suitable systems are controlled microcosm metapopulations with protozoa or bacteria in connected patches \cite{burkey1997metapopulation,Gonzalez1998,carrara2012}, in which repeated local extinction events can be monitored while independently varying effective local carrying capacity and dispersal connectivity, for instance $M$ could be set by patch volume or resource supply, and the distance from the persistence threshold $\Delta$ tuned by varying local extinction or colonization rates. Furthermore, combining the reduction to a single SDE with time-dependent boundaries could provide a natural framework for slowly varying environmental conditions \cite{Tuckwell1984,Molini_2011,BerLin2022,Bernardi_2024}.
In these settings, our theory predicts finite-size scaling of extinction statistics near the persistence threshold, as well as maximal variability at a finite rescaled distance from criticality $\hat{\Delta} > 0$, making the present framework useful not only as a theoretical prediction but also as a guide for interpreting extinction statistics in realistic ecological systems.}

\section{Interspecies interactions in competitive metacommunities}\label{sec:multispec-model}
\noindent
We now shift our focus from single-species metapopulation models in different settings to multi-species metacommunity models by incorporating interspecies interactions. As a starting point, we focus on the paradigmatic case of purely competitive species, since competition for space and resources is a pervasive and fundamental driver of ecological exclusion \cite{Tilman1994,DurLev1998,calcagno2006coexistence,Mesz_na_2006_competition,VioPu2010competition}.

We consider the same ecological reactions as in \cref{sec:micromodel}, and add a species index, using the convention that species are indicated by Greek indices and patches by Latin indices. Hence, a settled individual belonging to species $\alpha$ located within patch $i$ will be indicated with $S_{\alpha i}$. The reaction involving settlers introduced in \cref{sec:micromodel} naturally generalizes to 
\begin{equation}
	S_{\alpha i} \xrightarrow{e_{\alpha i}} \varnothing_i, 
	\label{eqn:reactions_settled_multispecies-0}
\end{equation}
where $e_{\alpha i}$ denotes the local death rate of species $\alpha$ in patch $i$. Reactions describing the generation of explorers are straightforward generalizations of the single-species case:
\begin{equation}
	S_{\alpha i} \xrightarrow{c_{\alpha i}  \, C_{\alpha, ij}} S_{\alpha i}  + X_{\alpha j}.
	\label{eqn:reaction_exbirth_multispecies}
\end{equation}
The explorers' reactions are also analogous to the single-species case: 
\begin{equation}
	\begin{gathered}
		X_{\alpha i} \xrightarrow{D \, W_{\alpha, ij}} X_{\alpha j} \\
		X_{\alpha i} \xrightarrow{\gamma_\alpha} \Phi \\
		\quad X_{\alpha i} + \varnothing_i \xrightarrow{\lambda_\alpha / M_i} S_{\alpha i} , \quad X_{\alpha  i} + S_{\alpha  i} \xrightarrow{\lambda_\alpha / M_i} S_{\alpha i}.
		\label{eq:reactions_explorers_multispecies}
	\end{gathered}
\end{equation}
An important difference with respect to the single-species case is that explorers of different species compete for the same empty spaces, which will lead to an interaction term in the effective metacommunity equations. 

For simplicity, we assume that $\gamma_\alpha =0$, $C_{\alpha,ij}=h_\alpha A_{\alpha,ij}$, and $W_{\alpha,ij}=A_{\alpha,ij}$. Under these assumptions, the set of microscopic reactions above leads to the following rate equations for the settler and explorer densities, which we indicate with $p_{\alpha i} = [{P_{\alpha i}}]/M_i$ and $x_{\alpha i} = [{X_{\alpha i}}]/M_i$, respectively:
\begin{align}
	\dot{p}_{\alpha i} &= -e_{\alpha i} p_{\alpha i} + \lambda_\alpha \left(1 - \sum_{\beta = 1}^S p_{\beta i}\right) x_{\alpha i} 	\label{eq:p-rate-equations-multi} \\		
	\dot{x}_{\alpha i} &= - \left(\lambda_\alpha x_{\alpha i} - h_{\alpha} \sum_{j = 1}^N A_{\alpha, ji} c_{\alpha j} p_{\alpha j}{\frac{M_j}{M_i}}\right) + D_\alpha \sum_{j = 1}^N\left(A_{\alpha, ji}x_{\alpha j}{\frac{M_j}{M_i}} - A_{\alpha, ij} x_{\alpha i}\right). \label{eq:x-rate-equations-multi}
\end{align}
As in \cref{sec:micromodel}, we assume that exploration occurs on time scales much shorter than those of the settled population; therefore, we set $\dot{x}_{\alpha i} = 0$. This quasi-stationary assumption leads to
\begin{equation}
	\lambda_\alpha x_{\alpha i} = h_\alpha \sum_{j,k = 1}^N \sum_{\beta = 1}^S (\hat{F}^{-1})^{\alpha \beta}_{ik} A_{\beta, jk}p_{\beta j}c_{\beta j}{\frac{M_j}{M_i}}
	\label{eq:QS-approx-multisp}
\end{equation}
where $\hat{F}$ is 
\begin{equation}
	F^{\alpha \beta}_{ij} = \delta_{\alpha\beta}\left[\delta_{ij} + f_\beta L_{\beta, ji}\right]    
\end{equation}
with $L_{\alpha, ij} = \delta_{ij}\sum_{k}A_{\alpha, ik} - A_{\alpha, ij}{M_i/M_j}$ the Laplacian matrix of the dispersal network of species $\alpha$, and $f_\alpha = D_\alpha/\lambda_\alpha$. As in the single-species case, the value of $f_\alpha$ quantifies the colonization efficiency of species $\alpha$. When $f_\alpha$ is greater than one, an explorer will, on average, diffuse farther through the network and reach more distant nodes from the exploration origin before attempting colonization. Conversely, when $f_\alpha$ is less than one, explorers will preferentially colonize patches located close to the origin.

\begin{figure}
	\centering
	\includegraphics[width=0.6\textwidth]{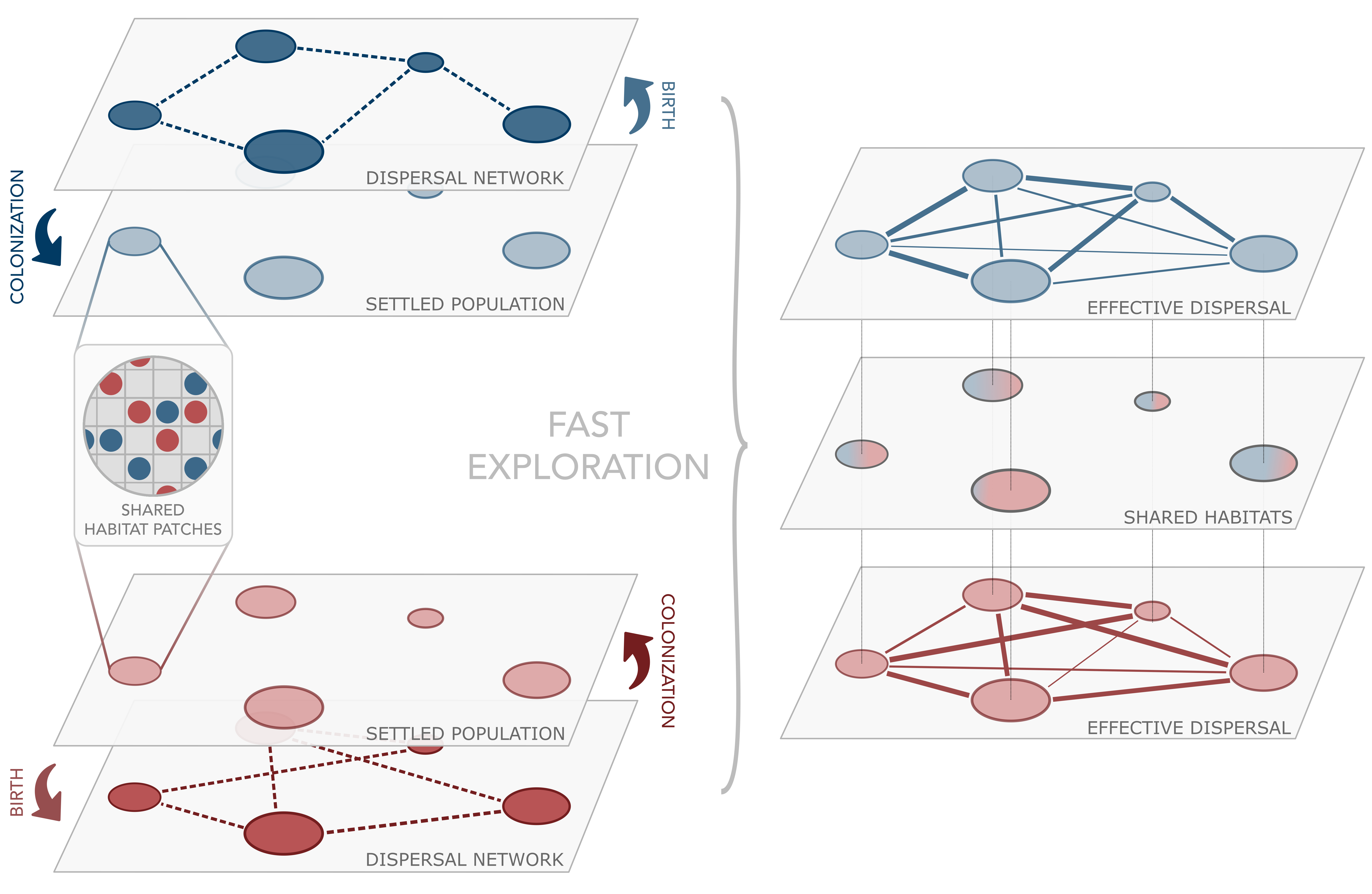}
	\caption{Visualization of the microscopic model with multiple species competing for space or a common resource leads to an effective metacommunity model. Assuming fast exploration, time-scale separation leads to an effective kernel description (\cref{eqn:kernel-multispec}) combined with the effective metacommunity equations reported in \cref{eqn:model-general}.}
	\label{fig:kernel-scheme-competition}
\end{figure}

Since $\hat{F}$ is block-diagonal with respect to the species index, we can compute its inverse species by species, obtaining
\begin{equation}
	(\hat{F}^{-1})^{\alpha \beta}_{ij} = \delta_{\alpha \beta}\sum_{k = 1}^N\frac{(\hat{V}_\alpha)_{ik}(\hat{V}^{-1}_\alpha)_{kj}}{1 + f_\alpha \omega_{\alpha k}}
		\label{eq:F-inverse-multisp}
\end{equation}
where $\omega_{\alpha k}$ is the $k$-th eigenvalue of $\hat{L}_{\alpha}$, and $\hat{V}_{\alpha}$ is the matrix of its right eigenvectors. Substituting  \cref{eq:F-inverse-multisp} into \cref{eq:QS-approx-multisp}, and the latter into the settled density equation \cref{eq:p-rate-equations-multi} leads to:
\begin{equation}
	\der{p_{\alpha i}}{t} = -e_{\alpha i} p_{\alpha i} + \left(1-\sum_{\beta=1}^{S}p_{\beta i}\right)\sum_{j=1}^{N} \,K_{\alpha, ij} \,p_{\alpha j},
	\label{eqn:model-general}
\end{equation}
where we introduced the dispersal kernel of species $\alpha$
\begin{equation}
	K_{\alpha, ji} = h_\alpha c_{\alpha i} {\frac{M_j}{M_i}}\sum_{k = 1}^N\sum_{l = 1}^N\frac{(\hat{V}_\alpha)_{ik}(\hat{V}^{-1}_\alpha)_{kl}}{1 + f_\alpha \omega_{\alpha k}}A_{\alpha, jl}\;.
	\label{eqn:kernel-multispec}
\end{equation}
We can set $h_\alpha = \xi_\alpha /(1 + f_\alpha^{-1})$, with $\xi_\alpha$ representing the maximum dispersal capacity of the species. With this formulation, exploration becomes impossible when $f_\alpha \to 0$, while the kernel stays finite as $f_\alpha \to \infty$. In general, different species need not share the same dispersal network. Note that this rigorous derivation clarifies the well-known extension of classical metapopulation models to metacommunities \cite{mouquet2002coexistence}. \new{We note that \cref{eqn:model-general} has the same structure as the single-species case discussed in \cref{sec:micromodel,sec:onespecies-theorems}, with the important difference that the ``excluded-volume" term $(1-\sum_{\beta=1}^{S}p_{\beta i})$ depends on all species, thus introducing interspecies interactions. Hence, species trying to colonize a patch experience competition from other species trying to colonize the same patch.}

\section{Coexistence and monodominance in homogeneous and heterogeneous habitats}\label{sec:multispec-results}
\noindent
As in the single-species case, we demonstrate how our framework can be applied to derive both analytical and numerical results on the coexistence of multiple species \new{that experience competitive interactions due to a shared common limiting resource, such as space}. We first consider a homogeneous habitat and landscape, showing that coexistence is possible only under fine-tuned conditions (a zero-measure set in parameter space). We then examine heterogeneous habitats within homogeneous landscapes (\ie, heterogeneity arising from local habitat properties) and derive analytical conditions for robust coexistence, identifying a critical level of habitat heterogeneity and its dependence on the metacommunity parameters. Finally, we show numerically that combining habitat and landscape (dispersal) heterogeneity further promotes coexistence. The proofs of all theorems and corollaries are in \cref{app:proofs-two}.

\subsection{Fragile coexistence in homogeneous habitats}
\noindent
We begin by examining the case of a homogeneous environment, in which local death, colonization, and birth rates are independent of the patch, and the colonization kernel is translation-invariant. As in \cref{sec:onespecies-theorems}, we set $c_{\alpha} = 1$. In this setting, species coexistence proves to be fragile, as established by the following theorem:
\begin{theorem}
	Consider the following metacommunity equations:
	\begin{equation}
		\der{p_{\alpha i}}{t} = -e_{\alpha} p_{\alpha i} + \left(1-\sum_{\beta=1}^{S}p_{\beta i}\right)\sum_{j=1}^{N}K_{\alpha, ij}\, p_{\alpha j},
		\label{eq:main-eq-multispec-TI-case}
	\end{equation}
and assume that each colonization kernel $\hat{K}_{\alpha}$ is translation-invariant, and $e_{\alpha}>0$ for all $\alpha$. Then, the only spatially uniform stationary state in which different species can coexist is a center manifold, and it occurs for a zero-measure set in the parameter space. The center manifold is (neutrally) stable.
\label{th:theorem-TI-multispec}
\end{theorem}

\noindent
The previous theorem shows that coexistence is possible in homogeneous environments when the ratio $e_\alpha / z_\alpha$ is equal for all species, where $z_\alpha=\sum_{j=1}^{N} K_{\alpha, ij}$ is the constant Fourier mode. Since $z_\alpha$ is proportional to the average dispersal kernel, we can interpret the ratio $e_\alpha/z_\alpha$ as the inverse of fitness. Hence, we define $r_\alpha=z_\alpha/e_\alpha$ as the (average) habitat-mediated fitness, reflecting how the effective dispersal kernel arises from the interplay between habitat heterogeneity and diffusion–colonization processes.

We now consider the case where species differ in fitness. According to \cref{th:theorem-TI-multispec}, such species cannot coexist. The following theorem establishes that the single species with the highest habitat-mediated fitness will ultimately be the only one to persist.
\begin{theorem}
	\label{th:monodominance}
	Consider the metacommunity equations in \cref{th:theorem-TI-multispec}, and further assume that $e_1/z_1 < e_\alpha/z_\alpha$ for all $\alpha \neq 1$ so that the first species has the largest habitat-mediated fitness. Furthermore, assume $e_1<z_1$, which is a necessary condition for the survival of the first species in the absence of other competing species ($z_1$ is the metapopulation capacity of the first species). Then, $p^*_{\alpha}= p_1^*\delta_{\alpha 1}$ and $p_{\beta i}^*=0$ for $\beta \neq 1$, the state in which only the first species survives and occupies all patches uniformly is a stable fixed-point solution of the dynamics.
\end{theorem}
\Cref{fig:frag-coex-selfc} provides a visual illustration of the two main results presented in this section. For clarity, we focus on the simplest case of two competing species. \Cref{fig:frag-coex-selfc}a shows the time evolution of all habitat patches, using different colors for the two species, when $r_1 = r_2$. In this case, \cref{th:theorem-TI-multispec} predicts stable coexistence, with the coexistence manifold forming a center manifold. Indeed, the dynamics rapidly collapse to a spatially uniform state in which all patches of the same species have identical values, corresponding to the relaxation of spatial modes with negative eigenvalues discussed in the proof of \cref{th:theorem-TI-multispec}. The point of the center manifold to which the system converges depends on the initial condition. Because the leading eigenvalue on the center manifold is zero, the manifold is marginally stable. A marginally stable stationary coexistence state is observed in several other ecological models \cite{posfai2017metabolic, tikhonov2017collective, mouquet2002coexistence}.

\begin{figure}
	\centering
	\includegraphics[width=\textwidth]{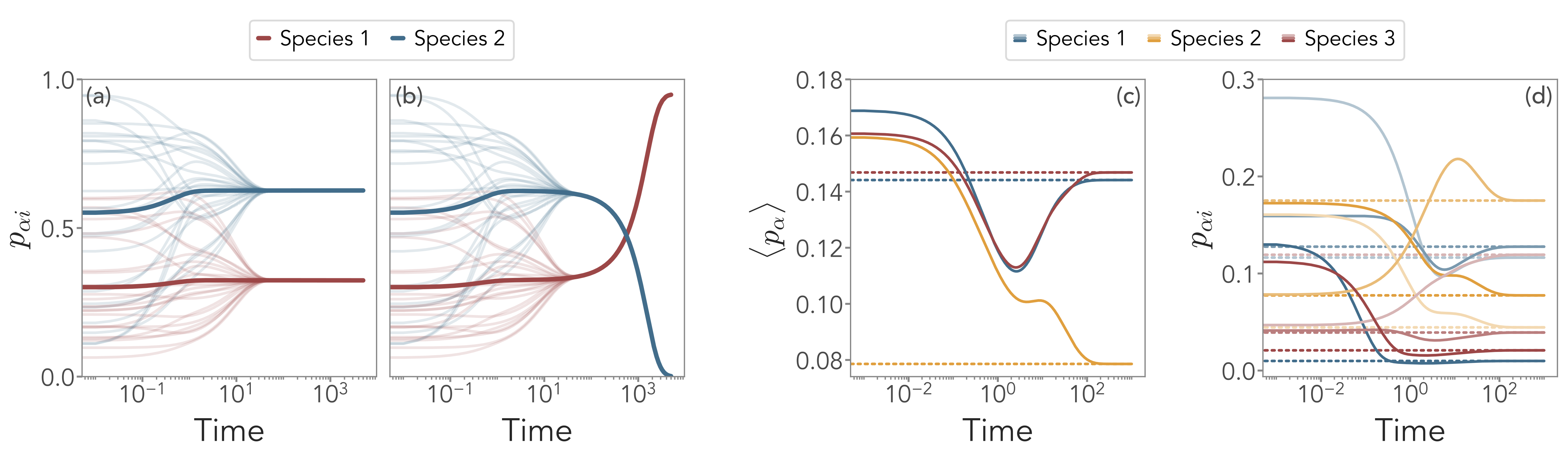}
	\caption{(a) Fine-tuned fragile coexistence in homogeneous environments. Time evolution of $p_{\alpha i}$ (the light colored lines correspond to different patches while the dark colored lines correspond to the averages across all patches, $\langle p_\alpha\rangle$), $i=1,\dots, N=25$ and $\alpha=1,2=S$ in a translationally invariant homogeneous environment (the kernel is computed from a 2D regular grid). The two species differ in exploration efficiency, with $f_1 = 1$ and $f_2 = 3$. a) Species’ (average) extinction rates are chosen such that $r_1 = e_1 / z_1 = e_2 / z_2 = r_2=0.002$, with $e_1 = 0.1, e_2=0.15$. These conditions satisfy the hypotheses of \cref{th:theorem-TI-multispec}, and the metacommunity dynamics converge to the center manifold.
    (b) Here, $r_2$ is $1\%$ smaller than in (a), leading to the extinction of species 2. The transient dynamics remain close to the center manifold, which acts as a slow manifold. Since the point on the center manifold reached by the trajectory depends on the initial condition, species 2 exhibits higher abundance during the transient phase, even though it ultimately declines to zero. 
    (c) Convergence of the dynamical evolution of the average $\ev{p_\alpha}$ to the fixed-point solution, with three species (dotted lines) given by \cref{eqn:Self-consistent-stationary} for the kernel given by \cref{eqn:MF-scaling}.
    (d) Same, but for the average abundance of each species in each patch, given by \cref{eqn:SS-general}. Here, $r_{\alpha i}$ were chosen from a lognormal distribution with mean 1.5 and standard deviation equal to 2.}
	\label{fig:frag-coex-selfc}
\end{figure}

\Cref{fig:frag-coex-selfc}b shows what happens when the metacommunity is initialized with the same conditions and parameters, except for a slightly reduced habitat-mediated fitness parameter in species two. The initial transient is nearly identical to the previous case, and the system temporarily settles into the same homogeneous state. However, the homogeneous coexistence solution subsequently decays, with species two eventually going extinct and species one taking over, as prescribed by \cref{th:monodominance}. Notably, the transient homogeneous coexistence state relaxes slowly and depends on the initial condition, as it lies close to the zero eigenvalue of the previous scenario.

\subsection{Robust coexistence and niche formation in heterogeneous habitats}
\noindent
To understand how habitat heterogeneity shapes ecosystem diversity, we now consider the general setting in which extinction rates vary across space: $e_{\alpha i}$ depends on both species $\alpha$ and the habitat patch $i$. The starting point is the full metapopulation model,
\begin{equation}
	\dot{p}_{\alpha i}  = -e_{\alpha i} p_{\alpha i} + \left(1 - \sum_{\beta=1}^{S} p_{\beta i} \right) \sum_{j = 1}^N K_{\alpha, ij} p_{\alpha j}.
	\label{eqn:SI:model-general-MS}
\end{equation}
A direct substitution shows that, in this more general scenario, the spatially uniform coexistence solution discussed in the previous section is no longer a stationary state of \cref{eqn:SI:model-general-MS}.

To derive explicit and minimal analytical conditions for coexistence, we first examine the mean-field limit, in which dispersal depends only on the species and not on specific patches. To ensure a meaningful limit for $N\to \infty$, we let the kernel elements scale as
\begin{equation}
	K_{\alpha, ij} = \frac{K_{\alpha}}{N}.
	\label{eqn:MF-scaling}
\end{equation}
which corresponds to a fully connected network in \cref{eqn:kernel-multispec}. We define the spatial average abundance as
\begin{equation}
	\frac{1}{N} \sum_{j} p_{\alpha j} = \langle p_{\alpha} \rangle,
\end{equation}
and the free space in patch $i$ as
\begin{equation}
	Q_i = 1 - \sum_{\beta=1}^S p_{\beta i}.
	\label{eqn:Qi-def}
\end{equation}
We can now write the following equation for the steady-state solution of \cref{eqn:SI:model-general-MS}:
\begin{equation}
	p^*_{\alpha i} = {r_{\alpha i} Q_i \langle p^*_{\alpha} \rangle},
	\label{eqn:SS-general}
\end{equation}
where $Q_i$ is implicitly determined by all $p^*_{\alpha i}$, and we define
\begin{equation}\label{eqn:fit-def}
	r_{\alpha i} = \frac{K_{\alpha}}{e_{\alpha i}}
\end{equation}
as the \emph{local} habitat-mediated fitness of species $\alpha$ in patch $i$, as $r_{\alpha i}$ quantifies the balance between colonization and extinction within each patch, and results from the interplay of species traits with local landscape properties. 
Using \cref{eqn:Qi-def}, \cref{eqn:SS-general}, averaging over patches yields
\begin{equation}\label{eqn:Q}
    Q_i= \left(1+ \sum_\alpha r_{\alpha i} \langle p^*_\alpha\rangle\right)^{-1}.
\end{equation}
and, using again \cref{eqn:SS-general}, we obtain the self-consistency equation:
\begin{equation}
	\langle p^*_\alpha \rangle = \frac{\lr{p^*_{\alpha}}}{N} \sum_{i=1}^N r_{\alpha i} \left( 1 + \sum_{\beta=1}^S r_{\beta i} \langle p^*_\beta \rangle \right)^{-1}.
	\label{eqn:Self-consistent-stationary}
\end{equation}
For any fixed choice of ${r_{\alpha i}}$, \cref{eqn:Self-consistent-stationary} can be solved numerically for $\langle \vec{p}\rangle$, which in turn allows the reconstruction of the full solution $p_{\alpha i}$ $\forall \alpha,i$ by exploiting \cref{eqn:SS-general}. 
\Cref{fig:frag-coex-selfc}c shows an example of how $\lr{p_\alpha}(t)$ (solid lines) evolves toward the self-consistent solution (dotted lines) of \cref{eqn:Self-consistent-stationary}, while \Cref{fig:frag-coex-selfc}d shows the same for individual patches, \cref{eqn:SS-general}.

To make analytical progress, we now treat these fitness values as quenched random variables. For simplicity, we drop the asterisk from now on, \ie $p^*_{\alpha i} \equiv p_{\alpha i}$. We focus on coexistence solutions, defined by
\begin{equation}
	\langle p_\alpha \rangle > 0 \quad \forall \alpha = 1, \ldots, S.
\end{equation}
In the limit $N \to \infty$, \cref{eqn:Self-consistent-stationary} becomes the following self-consistency relations for the mean densities $\langle p_{\alpha}\rangle$ 
\begin{equation}
	1 = \frac{1}{N} \sum_{i=1}^N r_{\alpha i} \left( 1 + \sum_{\beta=1}^S r_{\beta i} \langle p^*_\beta \rangle \right)^{-1}=\left\langle \frac{r_\alpha  }{1 + \sum_{{\beta}=1}^S r_{\beta} \langle p^*_{\beta} \rangle} \right\rangle.
	\label{eq:Selfc-palpha}
\end{equation}
Next, we assume that for each species $\alpha$, the random variables $r_{\alpha i}$ are drawn independently from a distribution $P_r(r | \vec{\zeta}_{\alpha})$ whose parameters $\vec{\zeta}_{\alpha}$ may differ across species.
We now introduce the Laplace transform of this distribution as
\begin{equation}
	W_{\alpha}(\omega) = \int_0^\infty \mathrm{d}r \, P_r(r | \vec{\zeta}_{\alpha}) \, e^{-r \omega}. 
	\label{eqn:W-definition}
\end{equation}
Using this definition, together with the identity valid for $z>0$,
\begin{equation}
	\frac{1}{z} = \int_0^\infty \mathrm{d}\lambda\, e^{-z \lambda} ,
\end{equation}
we derive from \cref{eq:Selfc-palpha} the following implicit equations for the vector $\LR{\vec{p}}$
\begin{equation}
	1 = \int_0^\infty \mathrm{d}\lambda e^{-\lambda} \left(-\frac{\partial}{\partial \omega} \ln W_{\alpha}(\omega)\right)_{\omega = \lambda \langle p^*_{\alpha} \rangle}
	\prod_\beta {W_\beta}\left(\lambda {\langle p^*_\beta \rangle}\right).
\end{equation}

\subsection{Critical habitat heterogeneity for coexistence}
\noindent
We now consider the case that $S$ is large, and perform the rescaling $x_{\beta} = S \langle p_{\beta}\rangle$ and $\lambda = S z$. 
The self-consistent equations for the rescaled mean abundances become
\begin{equation}
	1= S \int_0^\infty  \mathrm{d}z e^{-S \bar{F}(z, \vec{x})} \left( -\frac{W'_{\alpha}(z x_{\alpha})}{W_\alpha(z x_{\alpha})} \right)
	\label{eq:SC-MF-finiteS}
\end{equation}
where
\begin{equation}
	\bar{F}(z, \vec{x}) = z - \frac{1}{S} \sum_{\beta=1}^S \ln W_\beta (z \cdot x_\beta).
	\label{eqn:Fbar-def}
\end{equation}
Taking the derivative of \cref{eqn:Fbar-def} with respect to $z$ shows that
$\partial_z \bar{F}(z, \vec{x})>0$, which implies that its minimum occurs at $z = 0$. Thus, we expect the integral to be dominated by the value around $z=0$ when $S$ is large. We can perform a Taylor expansion of \cref{eqn:Fbar-def}:
\begin{equation}
	\bar{F}(z, \vec{x}) = z \left(1 + \frac{1}{S}\sum_\beta x_\beta \ev{r_\beta}\right) - \frac{z^2}{2}\frac{1}{S}\sum_\beta v_\beta^2x_\beta^2 + \mathcal{O}(z^3)
\end{equation}
where $v_\beta^2 = \ev{r^2_\beta}-\ev{r_\beta}^2$ denotes the variance of $P_r(r|\zeta_{\alpha})$.
We now introduce the following two quantities, defined as averages over species,
\begin{equation}
\overline{x\ev{r}} = \lim_{S\to \infty} \frac{1}{S} \sum_{\beta = 1}^S x_\beta \ev{r_\beta}; \qquad
\overline{x^2v^2} = \lim_{S\to \infty} \frac{1}{S} \sum_{\beta = 1}^S x_\beta^2v_\beta^2.
\end{equation}
Inserting the last three equations into \cref{eq:SC-MF-finiteS}, we obtain
\begin{align}
	1 = \frac{\ev{r_\alpha}}{1 + \overline{x\ev{r}}} + \frac{1}{S}\frac{1}{(1 + \overline{x\ev{r}})^2}\left[\frac{\ev{r_\alpha}}{1 + \overline{x\ev{r}}}\overline{x^2v^2} - x_\alpha v_\alpha^2\right] + \mathcal{O}\left(\frac{1}{S^2}\right) \; .
	\label{eqn:coexistence-eq-expanded}
\end{align}
\Cref{eqn:coexistence-eq-expanded} implies that for coexistence to be possible $\ev{r_\alpha}$ must scale as follows with respect to $1/S$:
\begin{equation}
	\ev{r_\alpha} = R + \frac{\Delta_\alpha}{S} + \mathcal{O}\left(\frac{1}{S^2}\right),
	\label{eqn:r-alpha-ensemble}
\end{equation}
where $R$ and $\Delta_\alpha$ are functions of $\vec{\zeta}_\alpha$. In particular, if $R$ is the average or ``baseline" fitness of the metacommunity, each species' deviation from such baseline must scale as $1 /S$, and can hence be indicated as $\Delta_\alpha / S$.
Next, we plug \cref{eqn:r-alpha-ensemble} into \cref{eqn:coexistence-eq-expanded} and obtain
\begin{align}
	\label{eqn:SI:consistency-last}
	1 + \overline{x\ev{r}}  &= R + \frac{1}{S} \left[\Delta_\alpha + \frac{1}{1 + \overline{x\ev{r}}}\left(\frac{R}{1 + \overline{x\ev{r}}} \overline{v^2x^2} - x_\alpha v^2_\alpha\right)\right]+ \mathcal{O}\left(\frac{1}{S^2}\right)\\
	&= R + \frac{1}{S} \left[\Delta_\alpha + \frac{1}{R}\left(\overline{v^2x^2} - x_\alpha v^2_\alpha\right)\right]+ \mathcal{O}\left(\frac{1}{S^2}\right) \; .
\end{align}
We now observe that the right-hand side cannot depend on the species index. Therefore, the rescaled stationary population $x_\alpha = S \ev{p_\alpha}$ must be 
\begin{equation}
	x_\alpha = \frac{\Delta_\alpha - H}{v_\alpha^2}R + \mathcal{O}\left(\frac{1}{S}\right),
\end{equation}
where $H$ is a constant to be determined. \Cref{eqn:SI:consistency-last} implies that 
\begin{align}\label{eqn:R-1}
	R - 1 & = \overline{x\ev{r}} =  \frac{R}{S}\sum_\alpha x_\alpha = \frac{R^2}{S}\sum_\alpha \frac{\Delta_\alpha-H}{v_\alpha^2},
\end{align}
at the leading order in $1/S$. We can now solve for $H$;
\begin{equation}
	H = \left(\frac{1}{S}\sum_\alpha \frac{1}{v_\alpha^2}\right)^{-1}\left[\frac{1}{S}\sum_\alpha \frac{\Delta_\alpha}{v_\alpha^2} - \frac{R-1}{R^2}\right] = \left(\;\overline{\frac{1}{v^2}}\;\right)^{-1}\left[ \left(\,\overline{\frac{\Delta}{v^2}}\,\right) - \frac{R-1}{R^2}\right],
\end{equation}
where we recall that $\bar{(\cdot)} := \frac{1}{S}\sum_\alpha (\cdot)$ indicates an average over species. Survival of all species occurs if
\begin{equation}
	\Delta_\alpha > H \quad \forall \alpha.
\end{equation}
We can simplify this coexistence condition assuming that $v_\alpha^2 = v^2$, that is, the disorder variance is the same for all species. In this case, the threshold for survival $H$ takes the simple form:
\begin{equation}
	\label{eqn:SI:coex_condition_alpha}
	\Delta_\alpha > H = \bar{\Delta} - v^2 \frac{R-1}{R^2} \;.
\end{equation}
where the overbar indicates the usual average over species
\begin{equation}
	\bar{\Delta} = \frac{1}{S}\sum_\alpha \Delta_\alpha.
\end{equation}
From \cref{eqn:R-1} and \cref{eqn:SI:coex_condition_alpha} it follows that $R > 1$ and 
\begin{equation}
	\text{min}_\alpha \Delta_\alpha > \bar{\Delta} - v^2 \frac{R-1}{R^2}.
\end{equation}
We can rewrite this last condition as follows
\begin{equation}
v^2 > v^2_c := \gamma_\Delta \, \frac{R^2}{R - 1}, \quad \gamma_\Delta = \bar{\Delta} - \min \Delta_\alpha,
\label{eqn:critical-disorder-final}
\end{equation}
which defines a critical value for the disorder variance, $v_c$. 
The quantity $\gamma_\Delta$ is a measure of how heterogeneous the species' average fitnesses are. Hence, \cref{eqn:critical-disorder-final} predicts that a stronger habitat disorder is needed for the coexistence of species with larger differences in their mean fitness. On the contrary, if they are all equal, \ie $\ev{r} = R$, then 
\begin{equation}
	\gamma_\Delta = 0 \;\implies\; x_\alpha = \frac{R - 1}{R},
\end{equation}
so that $R>1$ is the only condition determining either full coexistence or the extinction of all species.

\begin{figure}
\includegraphics[width=\textwidth]{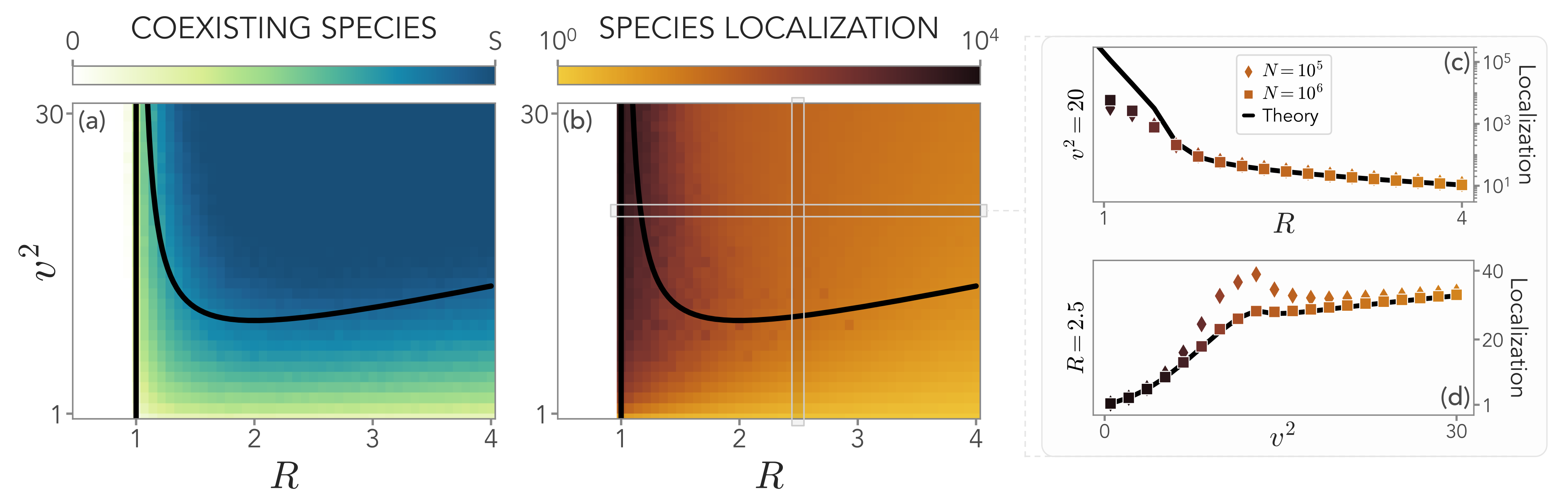}
\caption{Analytical phase diagram of coexistence regions with mean-field dispersion in heterogeneous habitats. (a) Fraction of coexisting species as a function of the baseline metacommunity fitness $R$ and habitat disorder variance $v^2$. (b) Species localization as a function of the same parameters shows increased localization in the coexisting phase and a strong increase near the global extinction boundary $R=1$. (c) Horizontal slice of panel (b) showing localization as a function of $R$, at a fixed disorder value. A strong increase in the localization near the extinction boundary is observed. (d) Vertical slice of panel (b), localization as a function of the disorder variance for a fixed $R$, which shows the increase of localization when transitioning from monodominance to coexistence, as an indication of niche formation. Numerical simulations are obtained using a log-normal distribution for the disorder. Here, $\Delta_\alpha$ is evenly spaced between $\pm 5/2$, so that $\gamma_\Delta = 2.5$. The theory is computed from numerical integration of \cref{eqn:moments-p-alpha}. An expansion for large $S$ only gives an accurate prediction for very large $S$. Adapted from ref.\ \cite{PadNic2024}.}
\label{fig:MF-coexistence}
\end{figure}

Now, we show how these analytical predictions compare to numerical simulations. In \Cref{fig:MF-coexistence}a, we choose a log-normal distribution for the local habitat-mediated fitness distribution $P_r(r)$ and show the fraction of coexisting species in the $(R, v^2)$-coordinate space for a fixed $\Delta_\alpha$. The black vertical line represents the first (co)existence condition $R>1$. Indeed, we observe that when $R<1$, no species survives. The second black line corresponds to the critical disorder variance boundary \cref{eqn:critical-disorder-final}. The transition from partial to full coexistence is smooth here because $S,N$ are finite, but becomes sharp in the limit $N,S\to \infty$. Coexistence of all species indeed requires sufficiently large habitat heterogeneity (i.e., large enough $v^2$). Interestingly, the critical variance increases as a function of $R$, as \cref{eqn:critical-disorder-final} makes clear. The intuitive interpretation is that, in this case, all species have high baseline fitness, which effectively increases competition. Hence, intermediate values of $R$ are optimal for coexistence under fixed habitat heterogeneity. This observation suggests that coexistence does not correspond to evenly spread species densities, but is rather associated with some kind of spatial localization.
To quantify localization explicitly, we use the following measure
\begin{equation}
	{\textrm{Loc}_{\alpha}} \equiv \frac{\frac{1}{N}\sum_i p_{\alpha i}^4}{\big(\frac{1}{N} \sum_i p_{\alpha i}^2 \big)^2} = \frac{\langle p_{\alpha}^4 \rangle}{\langle p_{\alpha}^2 \rangle^2},
	\label{eqn:localization-def}
\end{equation}
which is related to the inverse participation ratio. If the distribution of species $\alpha$ is uniform within a subset $E$ covering $K\leq N$ patches, then $p_{\alpha i} \approx \mathbf{1}_E(i) p^*$, where $\mathbf{1}_E$ is the indicator function of set $E$. Substituting into \cref{eqn:localization-def} we obtain
\begin{equation}
{\textrm{Loc}_{\alpha}} = \frac{N}{K}.
\end{equation}
Hence, when the species occupies all patches evenly $\textrm{Loc}_{\alpha} \approx 1$, and $\textrm{Loc}_{\alpha} \approx N$ when it is present only within a single patch.
To find an analytical expression for the localization, we need to compute the moments of $p_{\alpha}$. For this purpose, we combine \cref{eqn:SS-general} with \cref{eqn:Q}, yielding
\begin{equation}
	{p_{\alpha i}} = \frac{\LR{p_\alpha} r_{\alpha i}}{1+\sum\limits_{\beta=1}^S r_{\beta i}\lr{p_\beta}},
\end{equation}
which provides a link between the moments $r_{\alpha}$ and those of $p_\alpha$. Next, by employing the following identity, which holds for $z>0$,
\begin{equation}
	\frac{1}{z^m} = \frac{1}{\Gamma[m]} \int_0^\infty \mathrm{d}\lambda\, \lambda^{m-1} e^{-z \lambda},
\end{equation}
we can proceed as in the previous section to obtain
\begin{equation}
	\lr{p_{\alpha}^m}	= \frac{ x_\alpha^m}{\Gamma(m)} \int_0^\infty dz \, z^{m-1} e^{-S \bar{F}(z, x)} \frac{W^{(m)}(z x_{\alpha})}{W(z x_{\alpha})},
	\label{eqn:moments-p-alpha}
\end{equation}
where $x_\alpha = \lr{p_\alpha} S$, as above, and $W_{\alpha}(\omega)$ was defined in \cref{eqn:W-definition}. 

In \Cref{fig:MF-coexistence}b, we plot the localization, averaged over species, in the same $(R,v^2)$ coordinate space, while \Cref{fig:MF-coexistence}c-d show a horizontal and a vertical slice, respectively. The general behavior of the localization confirms the intuition that the coexisting phase is characterized by increased localization and, hence, reflects spatial niche formation. To show more explicitly that localization increases markedly as the habitat disorder $v^2$ grows and the system transitions from the monodominance to the coexistence phase, we consider a vertical cut of the phase space (\Cref{fig:MF-coexistence}d), where simulation data are compared to the theory for two values of the system size $N$. Another notable observation from \Cref{fig:MF-coexistence}b is that the localization undergoes a sharp increase as the global extinction boundary $R>1$ is approached. The ecological interpretation here is that near this threshold, species maximize their spatial segregation, as any competition could push them below the extinction threshold.

\begin{figure}
\includegraphics[width=\textwidth]{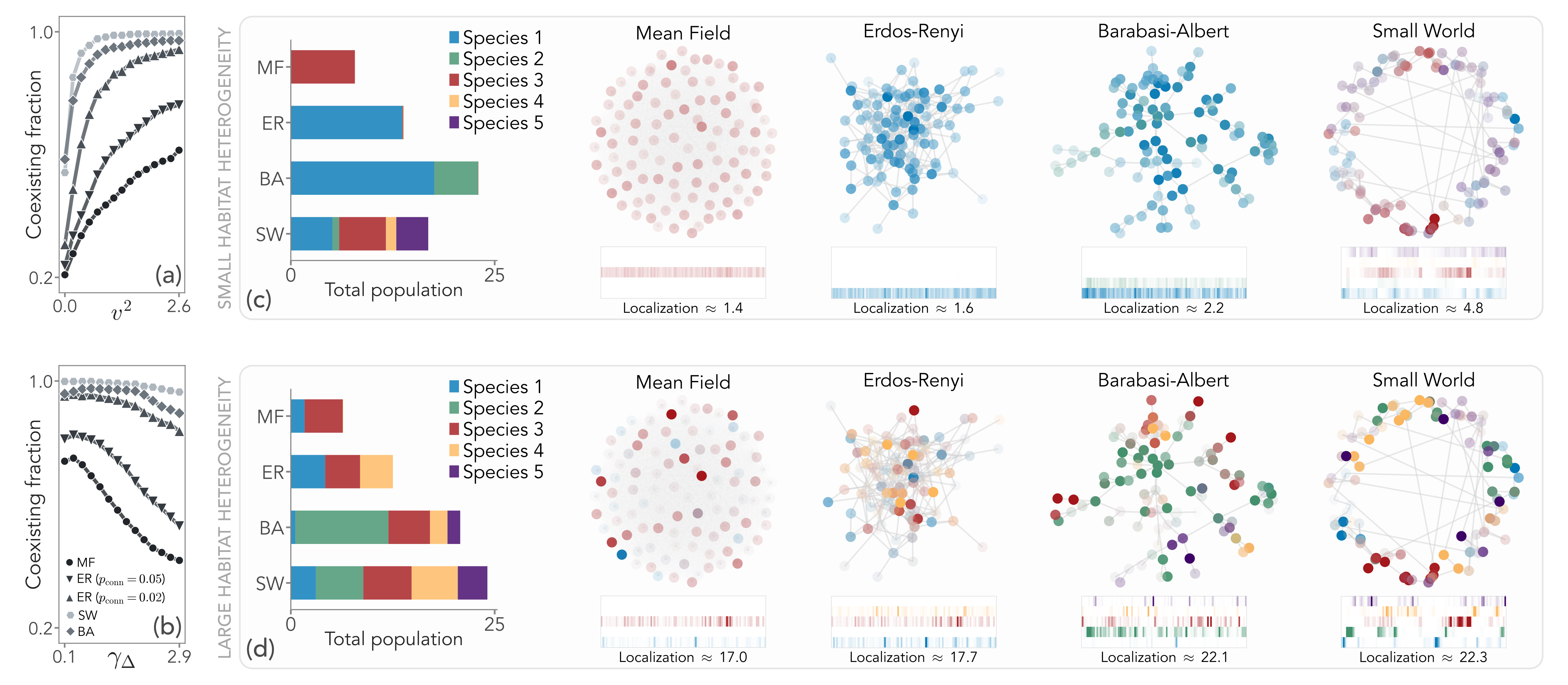}
\caption{Fraction of coexisting species as a function of habitat disorder (a) and species heterogeneity (b) for different network topologies: \ER (ER) with two different sparsity levels, small-world (SW), and Barabási-Albert (BA), compared to the mean-field (MF) kernel (see legend in panel b inset). In panel (a), $R = 2$, $\gamma_\Delta = 1$, $S = 5$, and $N = 100$. In panel (b), the same parameters are used with $v^2 = 3$. Due to finite-size effects, full coexistence is not observed in the mean-field case, even when $\gamma_\Delta$ is small. (c-d) At sufficiently high habitat heterogeneity, dispersal networks support a larger number of coexisting species and a higher total population than MF at constant average habitat-mediated fitness. Color transparency indicates the population density in each patch. Here, we set $\gamma_\Delta = 0$. Parts of figure adapted from \cite{PadNic2024}}
\label{fig:network-structure}
\end{figure}

\subsection{Landscape structure fosters coexistence}
\noindent
The mean-field kernel studied in the previous subsection corresponds to an all-to-all dispersal network. We can now compare the qualitative picture of the mean-field analytics to the more general scenario of a kernel obtained from a generic network using \cref{eqn:kernel-multispec}. More specifically, we contrast the mean-field (MF) results with three prototypical networks: \ER (ER) networks with two different average connection probabilities, thus illustrating the effects of network sparsity; Barabási-Albert (BA) networks, characterized by the presence of highly connected nodes; and small-world (SW) networks, which feature a high clustering coefficient. {To perform this comparison in practice, in analogy to \cref{eqn:fit-def}, we can define the local species fitness as $r_{\alpha i} = N \ev{K_\alpha}/e_{\alpha i}$, where $N \ev{K_\alpha}\equiv \sum_i K_{\alpha i}$, which reduces to \cref{eqn:MF-scaling} in the mean field case.} Then, we consider a realization of the quenched disorder $r_{\alpha i}$ and, for each species $\alpha$,  compute the kernel elements $K_{\alpha,ij}$ and its average $\ev{K_{\alpha}}$, from which we obtain the extinction rates as $e_{\alpha i} = N \ev{K_\alpha}/r_{\alpha i}$. In this way, $\ev{r_\alpha}$ does not depend on the specific network, ensuring comparable survivability for a given species $\alpha$ across different landscape topologies.
\Cref{fig:network-structure}a shows that the main finding remains qualitatively unchanged: the fraction of coexisting species increases with habitat heterogeneity $v^2$ across all networks, as in the MF case. However, the structure of the dispersal kernel quantitatively shifts the transition to coexistence, making it easier to achieve. ER networks support a larger fraction of coexisting species at a given disorder strength. Interestingly, coexistence is further enhanced as sparsity increases. At the same fixed level of sparsity, coexistence is also enhanced in BA networks and in a SW topology (\Cref{fig:network-structure}a). 

We can also observe our key result on the critical disorder strength \cref{eqn:critical-disorder-final} from a different perspective: at a fixed habitat disorder strength $v^2$ and baseline fitness of the community, there is a maximum level of species heterogeneity that can be supported. Indeed, \Cref{fig:network-structure}b shows that the fraction of coexisting species decreases with $\gamma_\Delta$ across all kernel types.

Finally, we visualize the spatial niche formation and the qualitative effect of kernel structure in \Cref{fig:network-structure}c. For simplicity,  we consider the ``neutral" case, in which $\gamma_\Delta = 0$ and the competing species have equal average habitat-mediated fitness.  For the same realization of $r_{\alpha i}$, structured landscapes support a larger number of species as well as a higher total population (\Cref{fig:network-structure}c). Interestingly, although all species in this example have the same average fitness, the interplay between spatial heterogeneity and the dispersal network structure determines which species survive where, leading to uneven population distributions (\Cref{fig:network-structure}c). For instance, species two (green) has the largest population in the BA network because it happens to have higher fitness in the major network hub. The clusters in the SW network appear to strike a balance between segregation, which buffers competition, and colonization boost, which increases population. Hence, landscape structure plays a key role in shaping both species survival and their spatial distribution.

\section{Outlook - Incorporating general ecological processes into dispersal dynamics}\label{sec:outlook}
\noindent
Until now, we focused on how species survival and coexistence are shaped by the interplay between dispersal across landscapes -- in different scenarios of spatial homogeneity and heterogeneity -- and competitive interactions with other species. We specifically examined competitive interactions, given their fundamental and widespread role in driving species exclusion from a metacommunity. However, a broad spectrum of ecological interactions can influence species survival within a metacommunity and affect overall ecosystem biodiversity. These processes include density-dependent birth and death rates, trophic interactions, or facilitation. Below, we give a few examples of how such processes can be incorporated into the present framework, focusing on the deterministic rate-based perspective.

One first extension is to allow for settled individuals to give birth directly, without going through the generation of explorers. {In the case of plants, this could represent vegetative propagation, which consists of new plants growing from vegetative parts such as roots, stems, or leaves.} In this scenario, we add the following reaction
\begin{equation}
	S_{\alpha i} + \varnothing_i \xrightarrow{b_{\alpha i}} 2~S_{\alpha i}.
	\label{eqn:reactions-settled-birthdeath}
\end{equation}
This reaction does not affect the explorer dynamics. In the fast-exploration limit discussed above, the effective kernel \cref{eqn:kernel-multispec}  also remains unchanged, and the metacommunity rate equations become
\begin{align}
	\dot{p}_{\alpha i}  &= - e_{\alpha i} p_{\alpha i} + b_{\alpha i}\left(1 - \sum_{\beta=1}^{S} p_{\beta i} \right) p_{\alpha i} + \left(1 - \sum_{\beta=1}^{S} p_{\beta i} \right) \sum_{j = 1}^N K_{\alpha, ij} p_{\alpha j}\\
	&= - (e_{\alpha i} - b_{\alpha i}) p_{\alpha i} -  \sum_{\beta=1}^{S} b_{\alpha i} p_{\alpha i} p_{\beta i} + \left(1 - \sum_{\beta=1}^{S} p_{\beta i} \right) \sum_{j = 1}^N K_{\alpha, ij} p_{\alpha j},
	\label{eqn:model-with-local-birth}
\end{align}
which makes it clear that the addition of this reaction contributes a logistic term, which can be split into a linear growth and a quadratic self-limitation term that depends on the local abundance of all species.

Another ecological reaction leading to a negative quadratic term  is the following:
\begin{equation}
	S_{\alpha i} + S_{\beta i} \xrightarrow{\Gamma_{\alpha \beta}} \varnothing_i + S_{\beta i},
	\label{eqn:reactions-extra-crowding}
\end{equation}
which models  generic competitive interference between different species (if $\alpha \neq \beta$) or an intra-species self-limitation term (if $\alpha = \beta$). \Cref{eqn:reactions-extra-crowding} models the death of a settled individual of species $\beta$ caused by the presence of a settled individual of species $\alpha$. Here, for notational simplicity, we assume that the competitive interference matrix $\Gamma_{\alpha \beta}$ does not depend on the patch index. This reaction leads to the following rate dynamics:
\begin{align}
	\dot{p}_{\alpha i}  &= - e_{\alpha i} p_{\alpha i} - \sum_{\beta=1}^{S}\Gamma_{\alpha \beta} p_{\alpha i}  p_{\beta i} + \left(1 - \sum_{\beta=1}^{S} p_{\beta i} \right) \sum_{j = 1}^N K_{\alpha, ij} p_{\alpha j},
	\label{eqn:model-with-competitive-interference}
\end{align}
which shows that this process also leads to a quadratic loss term. An equivalent rate description can be obtained by assuming that the death rate of settled individuals includes a density-dependent component:
\begin{equation}
	S_{\alpha i} \xrightarrow{ e_{\alpha i} +  p_{\beta i} \Gamma_{\alpha \beta}} \varnothing_i
	\label{eqn:dynamic-death-rate}.
\end{equation}
A quadratic self-interaction term (i.e. with only species-diagonal elements) can also arise as a simplified representation of the Janzen-Connell hypothesis \cite{Comita_2014}, that is, assuming that each species is affected by specific pathogens. These pathogens undergo logistic growth, tied to the abundance of the infected species:
\begin{equation}
\dot{\rho}_{\alpha i} =  \rho_{\alpha i} \left( p_{\alpha i}  - \frac{\Gamma_{\alpha \alpha}}{\kappa_\alpha} \rho_{\alpha i}\right),
\label{eqn:pathogen-1}
\end{equation}
where $\rho_{\alpha i}$ is the normalized density of pathogens. Assuming that the infection and recovery dynamics of the pathogen cycle are fast relative to the birth/death timescales of the settled population, we obtain from \cref{eqn:pathogen-1}:
\begin{equation}
	 p_{\alpha i}  =  \frac{\Gamma_{\alpha \alpha}}{\kappa_\alpha}  \rho_{\alpha i}.
	\label{eqn:pathogen-2}
\end{equation}
Introducing a pathogen-dependent death rate for species $\alpha$:
\begin{equation}
	S_{\alpha i} \xrightarrow{ e_{\alpha i} +  \rho_{\alpha i} \kappa_\alpha} \varnothing_i
	\label{eqn:dynamic-death-rate-pathogen}
\end{equation}
ultimately leads to:
\begin{align}
	\dot{p}_{\alpha i}  &= - e_{\alpha i} p_{\alpha i} - \Gamma_{\alpha \alpha} p_{\alpha i}^2 + \left(1 - \sum_{\beta=1}^{S} p_{\beta i} \right) \sum_{j = 1}^N K_{\alpha, ij} p_{\alpha j}.
	\label{eqn:model-with-pathogens}
\end{align}

Local interactions could also cause a species to be positively influenced by the presence of another, as in the case of facilitation or trophic interactions \cite{Bimler2025}.  In this scenario, the effective metacommunity description would be
\begin{align}
	\dot{p}_{\alpha i}  &= - e_{\alpha i} p_{\alpha i} +  \sum_{\beta=1}^{S}\Phi_{\alpha\beta}  p_{\alpha i}p_{\beta i}   +   \left(1 - \sum_{\beta=1}^{S} p_{\beta i} \right) \sum_{j = 1}^N K_{\alpha, ij} p_{\alpha j},
	\label{eqn:model-with-LV-interaction}
\end{align}
in which the interaction matrix $\Phi_{\alpha\beta}$ can have positive, negative, or zero entries according to the type of local interaction. To ensure that the local density does not exceed unity, the sign and magnitude of the $\Phi$ matrix entries need to obey specific constraints. For instance, the reaction
\begin{equation}
	S_{\alpha i} + S_{\beta i} \xrightarrow{\Phi_{\beta \alpha}}  2 S_{\alpha i},
	\label{eqn:reactions-extra-predation-alternate}
\end{equation}
which could be a simple model of a predator-prey  \cite{McKaneNewman2004_PRE}, implies that $\Phi_{\alpha\beta}=-\Phi_{\beta\alpha}$. More generally, requiring that, whenever $\Phi_{\alpha\beta} > 0$, the corresponding entry $\Phi_{\beta\alpha}$ is negative and at least as large in magnitude is sufficient to guarantee that the settled densities remain bounded by one.

\section{Concluding remarks}\label{sec:conclusion}
\noindent
This work refines and unifies recent developments in the mechanistic theory of metacommunity dynamics \cite{NicPad2023,PadNic2024,DoiNic2025}, providing a systematic framework to derive effective models that explicitly link microscopic ecological processes, landscape structure, and emergent community patterns. 

By starting from a microscopic description and coarse-graining to effective equations, we have shown how a wide range of ecological processes -- dispersal, colonization, extinction, stochasticity, and interspecific interactions -- can be incorporated within a single, analytically tractable formalism. \new{The derivation of the effective dispersal kernel that enables such compact description rests on the hypothesis of time-scale separation between dispersal processes and birth / death dynamics. Examples for which this separation is rather clean-cut include terrestrial plants and forests, in which dispersal occurs through fast vectors such as wind or animal-mediated seed transport relative to the long lifetimes of established individuals \cite{Nathan2000}; marine invertebrates, where sessile adults (e.g. corals or mussels) are coupled by rapidly dispersing planktonic larvae \cite{Cowen2009}; and fungi, where persistent local colonies are embedded in a continuous background of airborne spore dispersal \cite{Norros2014}. A similar, albeit less sharp, separation can emerge after spatial coarse-graining in terrestrial animals such as butterflies or resident birds, where local populations are relatively stable while dispersal events are infrequent but spatially extensive \cite{Hanski2003,Fuirst2021}.
Conversely, this assumption is expected to break down in systems where movement and demographic processes occur on comparable time scales. This is notably the case for many microbial communities, where growth, death, and dispersal are tightly coupled at the local scale, leading to range expansions and spatial dynamics that are driven by the same processes governing population turnover \cite{Hallatschek2010}. However, externally driven transport (e.g. fluid flow or host-mediated dispersal) can reintroduce an effective separation at larger scales. More generally, highly mobile organisms without clear site fidelity, or systems lacking a distinction between dispersal and reproductive stages, fall outside the regime in which a simple effective dispersal kernel provides an accurate coarse-grained description.}

In the single-species case, the structure of the dispersal network and the heterogeneity of local habitats jointly determine species persistence. In particular, \cref{th:theorem1} demonstrates how dispersal couples local habitat variations into a global property that governs persistence. This formal result aligns with empirical observations suggesting that local habitat heterogeneity can impact the global population status through dispersal \cite{SirHam2004heterogeneoushunting,RamTie2007heterogeneousmicrobialsystem,ZhangDeAngelis2020yeastheterogeneous}.

In the stochastic regime, we have connected extinction-time statistics and finite-size scaling exponents to underlying ecological parameters, revealing a universal scaling structure for noise-driven extinction. Future investigations could explore how persistence-time distributions and associated statistics change with variation in the microscopic reactions and/or dispersal-network properties, for instance by applying tools from response theory and stochastic inference in systems with absorbing boundary conditions \cite{Gandhi_1998, Yaari_shnerb2012persistence_scaling, Ovaskainen_2010, BerLin2020, padmanabha2023generalisation, kessler_schnerb2023, Su_2023, Keidar2024, AguMunAza2025inference}.

Extending the framework to multiple interacting species, we find that space effectively acts as the limiting resource: in homogeneous environments, competitive exclusion leads to monodominance except for a fine-tuned neutral manifold, in which the homogeneous coexistence state depends on the initial condition. Interestingly, we observed that a small amount of disorder does not support biodiversity, and we showed that a critical level of habitat heterogeneity is needed to create stable ecological niches and enable coexistence. This critical transition, smoothed for finite systems, might be a way of reconciling contrasting reports of the effect of habitat heterogeneity on species coexistence \cite{tamme2010environmental,Allouche_2012,Hortal_2013,Carnicer_2013,Hamm_2017,herberich2023environmental,MillerAllesina2023}.

The generality of the microscopic construction makes it possible to extend the present framework to other ecological processes. Future efforts should investigate the effect of the additional terms on species coexistence. Quadratic self-limitation terms are expected to introduce a self-regulation that buffers interspecific competition and could stabilize coexistence \cite{ktw2017goldenfeld}. Incorporating predator-prey or trophic interactions could likewise alter the coexistence conditions \cite{Kang_2024stoch_encounter} and possibly cause oscillations \cite{McKane2005}, chaotic population dynamics \cite{roy2020complex, pearce2020stabilization,altieri2022effects,mallmin2024chaotic}, offering a route to explore the emergence of non-equilibrium behavior in structured landscapes beyond the fixed-point equilibria observed here. 

A key challenge for future work lies in connecting theoretical parameters to empirical data. While the dispersal network might be inferred from spatial connectivity or movement data, detailed local rates and interaction parameters may be difficult to estimate. Recent work explored a promising approach assuming a suitable random kernel and inferring distribution parameters from empirical abundances in data-rich communities \cite{BerNic2026}. Altogether, this framework establishes a versatile and extensible foundation for studying spatially structured ecosystems through the lens of non-equilibrium statistical mechanics, enabling both theoretical advances and data-driven applications.

\bibliography{jstat-metapop}

%\clearpage

\appendix
\section{Proofs of theorems of \cref{sec:onespecies-theorems}}
\label{app:proofs-one}

\begin{proof}[Proof of \cref{th:theorem1}]
    We begin by linearizing \cref{eq:model-one-spec-lean} around the extinction solution $\vec{p}=0$, which yields 
	\begin{equation}
		\label{eq:linearized-extinction-1S}
		\dot p_i = -e_i p_i + \sum_{j=1}^N K_{ij} c_j p_j.
	\end{equation}
	According to the Perron-Frobenius theorem, the eigenvalue of the generalized landscape matrix $\hat{\mathcal{M}} = \hat{K}\hat{\mathcal{C}}\hat{\mathcal{E}}^{-1}$ with the greatest modulus is real and positive, denoted by $\lambda_\mathrm{max}$, and we can associate with it a left eigenvector $\vec w$ with strictly positive entries.  We now consider the temporal evolution of $\vec{w}^T\vec{p}(t)\equiv f(t)$, the projection of $\vec{p}$  onto  $\vec{w}$. In the remainder of this proof, we drop the arrow symbols from vectors to reduce notation clutter.
	Since $w$ is a left eigenvector of $\hat{\mathcal{M}}$ corresponding to the eigenvalue $\lambda$, we can write
	\begin{equation}
		w^T\hat{\mathcal{M}} = \lambda_\mathrm{max} w^T \quad \text{or, equivalently,} \quad
		w^T\hat{K}\hat{\mathcal{C}} = \lambda_\mathrm{max}  w^T\hat{\mathcal{E}}.
		\label{eq:eigenvectors}
	\end{equation}
	We can make use of this last relation to obtain the following equation for the time derivative of $f(t)$
	\begin{equation}\label{eq:fdot}
		\dot f= -w^T\hat{\mathcal{E}} p + w^T \hat{K}\hat{\mathcal{C}} p = w^T\hat{\mathcal{E}} p(-1 + \lambda_\mathrm{max}),
	\end{equation}
	We now define $e_\mathrm{min}>0$ as the minimum local extinction rate, \ie $e_\mathrm{min}\equiv \min\{e_i: i=1,\dots , N\} $, and derive the following inequality:
	\begin{equation}
		w^T\hat{\mathcal{E}} p=\sum_i w_i\, p_i\, e_i \geq e_\mathrm{min}\sum_i w_i\, p_i= e_\mathrm{min} f.
	\end{equation}
	If $\lambda_\mathrm{max}>1$, multiplying the last inequality by $\lambda_\mathrm{max}-1$ and combining with \cref{eq:fdot} leads to:
	\begin{equation}
	    \dot f(t)> e_\mathrm{min}\,(\lambda_\mathrm{max}-1) f(t),
	\end{equation}
	which implies
	\begin{equation}
		f(t) \geq f(0) e^{e_\mathrm{min}\,(\lambda_\mathrm{max}-1)t}.  
	\end{equation}
	Thus $f(t)$ grows exponentially, which means at least one component $p_i(t)$ must increase, showing that the extinction state, $\vec p=0$, is unstable. Consider now the case $\lambda_\mathrm{max}< 1$. From \cref{eq:model-one-spec-lean} one derives the following exact inequality , 
    \begin{equation}
        \dot p_i(t)\leq -e_i p_i + \sum_{j=1}^N K_{ij}c_jp_j.
    \end{equation}
    Using the same steps as above we get
    \begin{equation}
        \dot f \leq w^T\hat {\mathcal{E}} p(\lambda_\mathrm{max}-1)\leq e_\mathrm{min}(\lambda_\mathrm{max}-1) f,
    \end{equation}
    which implies 
    \begin{equation}
        f(t) \leq f(0) e^{e_\mathrm{min}\,(\lambda_\mathrm{max}-1)t}.
    \end{equation}
    The exponential decay of $f(t) = \sum_i w_i p_i(t)$ -- a sum of non negative terms -- implies that $p_i(t) \to 0$ for all $i$ as $t \to \infty$ since $w_i>0\,\, \forall i$. Hence, the extinction state, $\vec p=0$, is globally stable.
\end{proof}

\begin{proof}[Proof of \cref{th:firstcorollary}]
	Let $\vec{p^*} \neq 0$ indicate a non-trivial stable stationary state of \cref{eq:model-one-spec-lean}. Suppose $p_k^*=0$ for an index $k$. Then, \cref{eq:model-one-spec-lean} implies that
	\begin{equation}
		\sum_{j=1}^{N}K_{kj} c_j p_j^*=0.
		\label{eq:absurdum}
	\end{equation}
	This implies that $p^*_j=0$ for all $j$ such that $K_{kj} c_j>0$. Because $\hat{K}$ is irreducible, a positive integer $m(i)$ exists such that $(\hat{K}^m)_{ki}>0$ for each $i$. Consequently, for each $i$ there is at least one sequence of $n$ elements of $\hat{K}$ such that $K_{k k_{n-1}}K_{k_{n-1} k_{n-2}}\dots K_{k_1 i}>0$, $1\leq n\leq m$, with $k,\, k_{n-1},\,\dots ,\, k_1,\,i$ being all distinct. Therefore, $p_{k_{n-1}}^*=0$. We can now state \cref{eq:absurdum} with $p_{k_{n-1}}^*$ in place of $p_{k}^*$, which implies $p_{k_{n-2}}^*=0$. We can do this recursively, obtaining $p_i^*=0$. Since $i$ is arbitrary we find that $\vec {p^*}=0$, which contradicts the hypothesis that $\vec {p^*}$ is non-trivial.	
\end{proof}

\begin{proof}[Proof of \cref{th:HO-assumptions}]
	We begin by noting that, if $c_i = c A_i$, $e_i = e / A_i$, and $K_{ij} = e^{-\alpha d_{ij}}$ the generalized landscape matrix can be written as
    \begin{equation}
        \mathcal{M}_{ij} = \frac{c}{e} \mathcal{M}_{ij}^H, \quad \mathcal{M}_{ij}^H  \equiv A_ie^{-\alpha d_{ij}}A_j
    \end{equation}
    where $\mathcal{M}_{ij}^H$ is symmetric and is equivalent to the landscape matrix introduced by Hanski and Ovaskainen \cite{HanOva2000}. Thus, since $\delta = e/c$, the ``leading eigenvalue" or ``metapopulation capacity" of $\hat{\mathcal{M}}^H$ is $\lambda_M=\delta \cdot \lambda_\mathrm{max}$. Hence, the condition $\lambda_\mathrm{max}>1$ is equivalent to $\lambda_M>\delta$.
\end{proof}

\begin{proof}[Proof of \cref{th:theorem-TI-singlespecies}]
The matrix $\hat{K}\hat{\mathcal{E}}^{-1}$ is irreducible.  Let $\tilde{K}_0$ be the row sum of $\hat{K}$, which is row-independent. Then, $\vec{v}=(1,\dots 1)^T$ is a right eigenvector of $\hat{K}$  corresponding to the eigenvalue $\tilde{K}_0$. For non-negative matrices, the spectral radius is bounded by the minmum and maximum row sums from above and below, respectively. Hence, we conclude that $\vec{v}$ is the Perron-Frobenius eigenvector of both $\hat{K}$ and $\hat{K}\hat{\mathcal{E}}^{-1}$, since $\hat{\mathcal{E}}$ is a scalar matrix. Hence, from \cref{th:theorem1} we conclude that: if $\tilde{K}_0/e<1$ the absorbing state $p^*_x=0$ is globally stable; if $\tilde{K}_0/e>1$ the absorbing state is unstable, and $p_x^*=1-e/\tilde{ K}_0$ is a stationary state. 

To show that $p_x^*=1-e/\tilde{ K}_0$ is the only stationary state, assume that $s_x>0$ is a non-trivial steady state of the dynamics \cref{eq:model-one-spec-const-row-sum}. Then, we can write
\begin{equation}
\frac{1}{s_x} = 1 + \frac{e}{\sum K_{xy}s_y}.
\label{eq:proof-th-uniform}
\end{equation}
Define now $P=\min{s_x}$ and $Q=\max{s_x}$. Clearly, $0 < P \leq Q$, and
\begin{equation}
1 + \frac{e}{\tilde{K}_0 Q} \leq \frac{1}{s_x}  \leq 1 + \frac{e}{\tilde{K}_0 P}.
	\label{eq:proof-th-uniform-1}
\end{equation}
Since the first inequality must hold for all $x$, we can evaluate it for the index for which $s_{x_Q}=Q$, obtaining
\begin{equation}
	1 + \frac{e}{\tilde{K}_0 Q} \leq \frac{1}{s_{x_Q}} = \frac{1}{Q}.
	\label{eq:proof-th-uniform-1b}
\end{equation}
Doing the same for the second inequality in \cref{eq:proof-th-uniform-1} and combining with $1/Q \leq 1/P$ yields
\begin{equation}
1 + \frac{e}{\tilde{K}_0 Q} \leq \frac{1}{Q} \leq \frac{1}{P} \leq 1 + \frac{e}{\tilde{K}_0 P}.
	\label{eq:proof-th-uniform-2}
\end{equation}
The first and last inequalities imply
\begin{equation}
	Q \leq 1 - \frac{e}{\tilde{K}_0} \quad \textrm{and} \quad 1 - \frac{e}{\tilde{K}_0} \leq P,
	\label{eq:proof-th-uniform-3}
\end{equation}
respectively. The two last equalities are only consistent with $P\leq Q$ if $P=Q=s_x=1 - \frac{e}{\tilde{K}_0}$. This shows that  $p_x^*=1-e/\tilde{ K}_0$ is the only stationary solution. We now prove its linear stability by studying explicitly the eigenvalues of the Jacobian matrix  evaluated at $p_x^*=p^*=1-e/\tilde{ K}_0$
\begin{align}
J_{xy} = \frac{\partial \dot{p}_x}{ \partial  p_y}\Bigg|_{p^*} &= (-e + \tilde{ K}_0 p^*)\delta_{xy} + (1 - p^*) K_{xy}\\
&= \tilde{K}_0\delta_{xy} + \frac{e}{\tilde{K}_0} K_{xy}.
\end{align}
If $\lambda_\alpha$ indicate the eigenvalues of $\hat{K}$, the eigenvalues of the Jacobian matrix are $\frac{e}{\tilde{K}_0} \lambda_\alpha - \tilde{K}_0$. We have shown above that $\max_\alpha \Re{\lambda_\alpha} = \tilde{K}_0$. Hence, the eigenvalue with the maximum real part is $e - \tilde{K}_0<0$, which implies the linear stability of the equilibrium.
\end{proof}

\begin{proof}[Alternative proof of \cref{th:theorem-TI-singlespecies} based on Fourier modes]
	We consider the uniform stationary state $p_i=p^*$ for all $i$, in which all patches have the same population. We consider the global extinction solution $p^*=0$ first, which is always a fixed point of the system \cref{eq:model-one-spec-const-row-sum}. We then considered the system's linearization around it:
	\begin{equation}
		\dot p_i = -e p_i + \sum_{j \in \Lambda}  K_{ij}p_j.
		\label{eq:one-spec-Tinv-pzero}
	\end{equation}
	{Because we dealing with a linear, translation invariant system, plain waves $\mathbf v_k=\left(e^{i 2\pi k \cdot j/ N }\right)_{j=0,1,\dots N-1}$ with $ k\in \{0, \dots, N-1\}^d$ represent a full set of eigenfunctions, or modes. Thus $\mathbf{p}(t)= \sum_{k=0}^{N-1}a_k(t)\mathbf{v}_k$ where the mode's amplitude response can be computed using \cref{eq:one-spec-Tinv-pzero}, leading to}
	\begin{equation}
		\dot{ a}_k = \big(-e  + {\tilde{ K}}_{k}\big)  a_k,
		\label{eq:one-spec-Tinv-zero-eigenf-eq}
	\end{equation}
	where $\tilde{ K}_k=\sum_j e^{-i 2\pi k j /N} K_{ij}$ indicates the discrete Fourier transform of the kernel $\hat{K}$. \Cref{eq:one-spec-Tinv-zero-eigenf-eq} is solved by an exponential function, which decays to zero when $-e  + \Re{{\tilde{ K}}_{k}} <0$. Requiring $\max\limits_{k}{\Re{\tilde{ K}_{k}}} <e$ guarantees then that the amplitude of all modes decays, and that the extinction equilibrium is thus stable. Since  $\Re{{\tilde{  K}}_{k}}=\sum_j K_{ij} \cos(2\pi k j /N) \leq \sum_j K_{ij} =\tilde{ K}_0$, then $\max\limits_{k}{\Re{{\tilde{ K}}_{k}}} =\tilde{K}_0$, we conclude that $\tilde K _0 < e$ implies the local stability of $p^*=0$, which proves the first part of the theorem. Naturally, this condition is consistent with \cref{th:theorem1}, because $\tilde{ K}_0 = \lambda_M$ is the largest eigenvalue of the dispersal kernel. 
	
	To prove the second part of the theorem, we begin with determining the nonzero uniform fixed point from \cref{eq:one-spec-Tinv-pzero}:
	\begin{equation}
		p^*=\frac{\tilde{K}_0-e}{\tilde{K}_0}.
		\label{eq:onesp-transinv-pstar}
	\end{equation}
	First, we notice that  \cref{eq:onesp-transinv-pstar} is unfeasible when $e>\tilde{K}_0$, i.e., when the trivial solution is stable. Conversely, the nontrivial solution is feasible when $e<\tilde{K}_0$. To study its stability, we compute the Jacobian matrix and evaluate it at $p_i=p^*$:
	\begin{equation}
		J_{\ell j}=\der{\dot{p}_\ell }{p_j}\Bigg|_{p^*} = \delta_{\ell j}\Big(-e-\tilde{K}_0 p^*\Big)+(1-p^*)K_{\ell j},
	\end{equation}
	which has the following eigenvalues
	\begin{equation}
		(-e-\tilde{K}_0 p^*) + (1-p^*) \tilde{K}_k.
	\end{equation}
	For the fixed point to be stable, all eigenvalues must have a negative real part. In other words, the largest real part must be negative, \ie 
	\begin{equation}
		0> \max_k\Big\{(-e-\tilde{K}_0 p^*) + (1-p^*) \tilde{K}_k\Big\} =-e-\tilde{K}_0p^* + (1-p^*) \tilde{K}_0 = e-\tilde{K}_0,
	\end{equation}
	which proves the second part of the theorem.
\end{proof}

\section{Proofs of theorems of \cref{sec:multispec-results}}\label{app:proofs-two}

\begin{proof}[Proof of \cref{th:theorem-TI-multispec}]
	We begin with noting that if $\hat{K}_{\alpha}$ is translation-invariant, then so is  $c_{\alpha} \hat{K}_{\alpha}$. Hence, we can consider the case $c_{\alpha}=1$, which amounts to a redefinition of $K_{\alpha,ij}$ that does not change the theorem's hypotheses. 
	We now consider the uniform stationary solution, \ie $p^*_{\alpha, i} = p^*_{\alpha}$ for all $i$, and introduce $Q \equiv 1-\sum_{\beta=1}^{S}p^*_{\beta}$, so that a uniform stationary solution to 
	\cref{eq:main-eq-multispec-TI-case} must satisfy
	\begin{equation}
		e_{\alpha} p^*_{\alpha} = Q \sum_{j=1}^{N} K_{\alpha, ij} p^*_{\alpha}.
		\label{eq:stationary-tras-inv}
	\end{equation}
	If species $\alpha$ survives, then $p^*_{\alpha}>0$, which implies
	\begin{equation}
		\frac{e_{\alpha}}{z_{\alpha}} = Q,
		\label{eq:coex-cond-CM}
	\end{equation}
	where we introduced the shorthand for the zero-frequency Fourier mode of $\hat{K}_\alpha$
	\begin{equation}
		z_{\alpha} \equiv  \sum_{j} K_{\alpha, ij}=\tilde{K}_{\alpha, 0},
		\label{eq:def-zalpha}
	\end{equation}
	and $Q \in (0,1)$. \Cref{eq:coex-cond-CM} must hold for any species that survives. Hence, given the set $E$, one of the $2^S$ possible subsets of the $S$ species, a spatially uniform solution in which all species within $E$ coexist is defined by:
	\begin{equation}
		\mathcal{C}=\{ \mathbf p^*: p_{\alpha i}\equiv p_\alpha^*\,  \textrm{if}\, \alpha \in E, \, p_{\alpha i}=0\, \textrm{if}\,\alpha \notin E\,\textrm{and}\, \sum_\alpha p_{\alpha}^* +Q-1=0\}
		\label{eq:centralmanifold}
	\end{equation}
	We now show that $\mathcal{C}$ is a center manifold. To this end, we first compute the system's Jacobian, $\hat{J}$, from  \cref{eq:main-eq-multispec-TI-case}, and evaluate it at a point belonging to $\mathcal{C}$:
	\begin{align}
		\der{\dot{p}_{\alpha i}}{p_{\beta k}}\Bigg |_{\mathbf p^* \in \mathcal{C}}=J_{(\alpha,i), (\beta,k)} &= -\delta_{ik} \delta_{\alpha \beta} e_{\alpha} + \Big(1 - \sum_{\nu} p^*_{\nu}\Big) K_{\alpha, ik} \delta_{\alpha \beta} - \delta_{ik} \sum_{\ell} K_{\alpha, i\ell} p^*_{\alpha} \label{eq:Jacobian}\\
		&= -\delta_{ik} \delta_{\alpha \beta} e_{\alpha} + Q K_{\alpha, ik} \delta_{\alpha \beta} - \delta_{ik}  p^*_{\alpha} z_{\alpha}.
	\end{align}	
	To prove that $\mathcal{C}$ is a center manifold, we consider the following set:
	\begin{equation}
		\mathcal{C}_0=\{ \mathbf  v: v_{\alpha i}\, \textrm{independent of}\, i\, \textrm{and}\, \sum_\alpha v_{\alpha i} =0\}.
	\end{equation}
	According to this definition, a vector $\mathbf v\in\mathcal{C}_0$ is uniform with respect to the patch index, and the sum over species index is null, \ie, $\sum_\alpha v_{\alpha i}= 0$.	
	Applying the Jacobian to $\mathbf{v} \in \mathcal{C}_0$ demonstrates directly that $\mathbf{v}$	is an eigenvector of $\hat{J}$ corresponding to the zero eigenvalue (we write $v_{\beta k}=v_\beta$ for simplicity, since the components of $\mathbf v$ cannot depend on the patch index):
	\begin{align}
		\sum_{\beta, k} J_{(\alpha,i), (\beta,k)} v_{\beta, k} &=
		-\sum_{\beta}\sum_{k} \delta_{\alpha\beta}\delta_{i k}e_{\alpha}v_\beta 
		+ Q \sum_{\beta}\sum_{k} K_{\alpha,i k}\delta_{\alpha\beta}v_\beta
		+ \sum_{\beta, k} p_{\alpha}^* z_{\alpha}v_\beta\notag \\
		&= v_\alpha(-e_\alpha +z_\alpha Q) -p_\alpha^* z_\alpha \sum_\beta v_\beta=0.\notag
		\label{eq:Jacobian-TI-ms-case}
	\end{align}
	We conclude that the Jacobian has a nontrivial null space, which implies that $\mathcal{C}$ is a center manifold.
	
	We now prove the stability of $\mathcal{C}$. To this end, we first show that there are no positive eigenvalues. To begin with, we note that the Jacobian has a block structure, in which each of the $S$ diagonal blocks, indexed by $\alpha$, is $ Q\hat{K}_\alpha - e_\alpha I_N + \hat p_\alpha$, where $\hat{p}_\alpha = -p^*_\alpha z_\alpha$. Each off-diagonal block is $I_N \hat{p}_\alpha$, where $\alpha$ is the row index, and off-diagonal blocks do not depend on the block column index.
	We define $U\in\mathbb{C}^{N\times N}$ as the basis change matrix that diagonalizes every block (a spatial plain wave basis), \ie each column of $U$ is $e^{i q \ell}$ and $q$ is the wave number.
	We can write for the entire block matrix
	\begin{equation}
		P =	\operatorname{diag}(U,\dots ,U) = I_{S}\otimes U \qquad \bigl(P^{-1}=P^{\dagger}\bigr).
	\end{equation}
	We apply this similarity transformation to the Jacobian matrix,
	\begin{equation}
		\tilde J \;:=\;P^{-1}JP
		\;=\;\bigl[\tilde J(q)\bigr]_{q=0}^{N-1},
	\end{equation}
	which splits the problem into the Fourier modes $\;\tilde J(q)\in\mathbb{C}^{S\times S}$ (the basis change has no effect on the off-diagonal blocks because they are scalar matrices).
	Explicitly, we find a Jacobian for each mode, which is
	\begin{equation}
		\tilde J_{\alpha\beta}(q)
		=\bigl(Q \lambda_{K_\alpha}(q)-e_{\alpha}+\hat p_{\alpha}\bigr)\,\delta_{\alpha\beta}
		\;+\;
		\hat p_{\alpha}\,(1-\delta_{\alpha\beta})
		\label{eq:Jacobian-Fourier}
	\end{equation}
	with block (species) indices $\alpha,\beta=1,\dots ,S$ and $\lambda_{K_\alpha}(q)$ is the $q$ eigenvalue from each $\hat{K}_\alpha$.
	We now set
	\begin{equation}
		d_{\alpha}(q):=Q \lambda_{K_\alpha}(q)-e_{\alpha},
		\qquad
		u_{\alpha}:=\hat p_{\alpha},
		\qquad
		\mathbf 1:=(1,\dots ,1)^{\!\top},
		\label{eq:def-four-jac-decomposition}
	\end{equation}
	so that we split the matrix into two parts:
	\begin{equation}
		\tilde J(q)=D(q)+u\,\mathbf 1^{\!\top},
		\label{eq:Jacobian-fourier-separation}
	\end{equation}
	where $D_{\alpha\beta}(q)=d_{\alpha}(q)\,\delta_{\alpha\beta}$ is diagonal and the other term has a rank one.
	We now look for the eigenvalues of $\tilde{J}(q)$, which are the zeros of the equation 
	\begin{equation}
		\det\bigl(\tilde J(q)-\mu I_{S}\bigr)
		=\det\bigl(D(q)-\mu I_{S}\bigr)
		\Bigl[\,1+\mathbf 1^{\!\top}(D(q)-\mu I_{S})^{-1}u\,\Bigr],
		\label{eq:eigenvalues-J-q}
	\end{equation}
	where we used the matrix determinant lemma.
	Since $D(q)$ is diagonal, the term in square brackets is easy to compute, which yields
	\begin{equation}
		\mathbf 1^{\!\top}(D(q)-\mu I_{S})^{-1}u
		=\sum_{\alpha=1}^{S}\frac{\hat p_{\alpha}}{d_{\alpha}(q)-\mu}.
	\end{equation}
	Hence, the characteristic equation to find the eigenvalues becomes
	\begin{equation}
		\det\bigl(\tilde J(q)-\mu I_{S}\bigr)
		=\prod_{\alpha=1}^{S}\bigl[d_{\alpha}(q)-\mu\bigr]\;
		f_{q}(\mu) = 0
		\label{eq:char-eq-final}
	\end{equation}
	where
	\begin{equation}
		f_{q}(\mu)
		:=1+\sum_{\alpha=1}^{S}
		\frac{\hat p_{\alpha}}{d_{\alpha}(q)-\mu}.
	\end{equation}
	Recalling that $\lambda_{K_\alpha}(0) = z_\alpha$, we now find for $q=0$ that
	\begin{equation}
		d_\alpha(0) = Q \lambda_{K_{\alpha}} (0) - e_\alpha = e_\alpha \left(\frac{\lambda_{K_{\alpha}}(0)}{z_\alpha} - 1 \right) = 0.
	\end{equation}
	Hence, the characteristic equation has a zero root with multiplicity $S-1$, and the only non-zero  root of the characteristic equation is 
	\begin{equation}
		f_{0}(\mu)=1+\sum_{\alpha=1}^{S} \frac{\hat p_{\alpha}}{-\mu} = 0 \rightarrow \bar{\mu}_0 = \sum_\alpha \hat{p}_\alpha = -\sum_\alpha z_\alpha p_\alpha^* < 0.
	\end{equation}
	Hence, the dominant eigenvalue for $q=0$ is zero. Recalling the definition of $d_\alpha(q)$, \cref{eq:def-four-jac-decomposition}, we now consider $q \neq 0$, and observe that
	\begin{equation}
		d_\alpha(q) < d_\alpha(0) = 0
		\label{eq:nonuniform-diagonal-are-neg}
	\end{equation}
	where the  inequality holds because $\hat{K}_\alpha$ is positive and translation invariant, a fact used also in \cref{th:theorem-TI-singlespecies}. Assuming now that species are different, then $d_\alpha(q)$ are distinct when $q \neq 0$. Direct substitution shows that none of $d_\alpha(q)$ are zeros of the characteristic equation. Hence, all non-uniform eigenvalues are roots of $f_q(\mu)$.  We now take the real part of $f_q(\mu)$
	\begin{equation}
		\Re{f_q(\mu)} = 1 + \sum_{\alpha=1}^{S} \hat p_{\alpha} \frac{\Re{d_\alpha(q)} - \Re{\mu}}{|d_{\alpha}(q)-\mu|^2},
	\end{equation}
	and assume that $\Re{\mu} \geq 0$. Then, using \cref{eq:nonuniform-diagonal-are-neg} and recalling that $\hat{p}_\alpha<0$ for all $\alpha$, we conclude that all terms in the sum are positive. Hence, the real part of $f_q(\mu)$ is always strictly larger than one, and cannot have any root in the $\Re{f_q(\mu)} \geq 0$ half plane when $q\neq 0$. We conclude that all eigenvalues for $q \neq 0$ have strictly negative real part. Hence, the maximum eigenvalue is zero and has with multiplicity $S-1$, which defines the center manifold introduced above.
	
	Finally, we compute the system's response to a small perturbation from one point belonging to the center manifold $\mathcal{C}$, keeping all nonlinear terms.
	We insert
	\begin{equation}
		p_{\alpha i} = p_\alpha^* + \Delta p_{\alpha i}
	\end{equation}
	into the equations of motion, obtaining
	\begin{align}
		\dot p_{\alpha i} = \Delta \dot  p_{\alpha i} &= -e_\alpha \left( p_\alpha^* + \Delta p_{\alpha i} \right) +
		\left[ 1 - \sum_{\beta} \left( p_{\beta}^{*} + \Delta p_{\beta i} \right) \right]
		\sum_{j} K_{\alpha, i j} \left( p_{\alpha}^{*} + \Delta p_{\alpha j} \right)	\\
		&= - e_{\alpha} p_{\alpha}^{*} - e_{\alpha} \Delta p_{\alpha i} 
		+ \underbrace{\left( 1 - \sum_{\beta} p_{\beta}^{*} \right)}_{\text{Q}} 
		\sum_{j} K_{\alpha i j} \left( p_{\alpha }^{*} + \Delta p_{\alpha j} \right) \\
		& - \sum_{\beta} \Delta p_{\beta i} \sum_{j} K_{\alpha, i j} 
		\left( p_{\alpha}^{*} + \Delta p_{\alpha j} \right) \\ &= 
		- e_{\alpha} p_{\alpha}^{*} - e_{\alpha} \Delta p_{\alpha i} 
		+ Q \underbrace{\sum_{j} K_{\alpha, i j} }_{z_{\alpha}} p_{\alpha}^{*}
		+ Q \sum_{j} K_{\alpha, i j} \Delta p_{\alpha j} \\
		& - \sum_{\beta} \Delta p_{\beta i} \sum_{j} K_{\alpha, i j} p_{\alpha}^{*} 
		+ \sum_{\beta} \Delta p_{\beta i} \sum_{j} K_{\alpha, i j} \Delta p_{\alpha j}	\\ 
		&= 	\sum_{\beta,j} J_{(\alpha,i),(\beta, j)} \Delta p_{\beta j} - \sum_{j} K_{\alpha, ij} \Delta p_{\alpha j} \sum_{\beta}\Delta p_{\beta i}.
	\end{align}
	Consistently, the constant term vanishes, the linear term coincides with the Jacobian, and the only nonlinear part is quadratic. 
	
	It is easy to see that the eigenvectors corresponding to the zero eigenvalue span the $S-1$ dimensional space defined by the $\mathcal{C}_0$ set defined above. For any $v \in \mathcal{C}_0$, the linear and quadratic responses are zero, and since there are no higher order terms, we conclude that any perturbation along this direction does not feel any response field, and does not leave $\mathcal{C}$. For perturbations along any other direction, the linear term dominates, and acts hence as a restoring force. 
\end{proof}

\begin{proof}[Proof of \cref{th:monodominance}]
	When $p_{\alpha}^* = p_1^* \delta_{\alpha 1}$, the Jacobian \cref{eq:Jacobian} is upper block triangular, so its eigenvalues are those of the diagonal blocks. 	
	The block corresponding to the surviving species $\alpha=1$ is $(-e_1 - z_1 p_1^*) \mathbb{I}_N + Q \hat{K}_1$, with $Q = e_1 / z_1$, $p_1^* = 1 - e_1 / z_1$, and $\hat{K}_1 = (K_{1,ij})$. Since $\hat{K}_\beta$ is positive for all $\beta$ (irreducible would suffice), its largest eigenvalue $\lambda_{M,1}$ is bounded between the minimum and maximum row sums. Translational invariance makes the row sum constant, $z_1$, so $\lambda_{M,1} = z_1$. The largest eigenvalue of this block is $(-e_1 - z_1(1 - e_1/z_1)) + (e_1/z_1) \cdot z_1 = e_1 - z_1 < 0$.
	For each extinct species $\alpha \neq 1$, \cref{eq:Jacobian} with $\alpha=\beta$ gives the block $-e_\alpha \mathbb{I}_N + Q \hat{K}_\alpha$, where again $Q = e_1 / z_1$. Its largest eigenvalue is $-e_\alpha + (e_1/z_1) z_\alpha < -e_\alpha + (e_\alpha/z_\alpha) z_\alpha = 0$, where the inequality follows from the hypothesis $e_1 / z_1 < e_\alpha / z_\alpha$ for all $\alpha \neq 1$.
	Since all blocks have strictly negative leading eigenvalues, the state $p_{\alpha}^* = p_1^* \delta_{\alpha 1}$ is locally stable.
\end{proof}

\section{Derivation of stochastic dynamics}\label{app:stochmetapop}
\noindent
We refer to Ref.~\cite{DoiNic2025} for a detailed derivation of the reduced SDE. Briefly, we start from the two-dimensional Fokker–Planck equation \eqref{eq:FP} for the joint probability density $\mathcal{P}(\vec{\rho},\vec{x},t)$, we derive the one-dimensional stochastic dynamics for the effective settled population density $\rho$. 
To reduce dimensionality, we begin by assuming a fully connected dispersal network and homogeneous habitat patches ($c_i = c$, $e_i = e$ $\forall\, i$),  rendering the system translationally invariant. Setting $\mathcal{A} = 1/N$ ensures a well-defined large-$N$ limit. 
By integrating out all sites except one, we define the marginalized probability distribution $P(\rho,x,t)$. 
Under mean-field–like approximations consistent with the fully connected topology, each variable is replaced by its effective value. 
In the limit $N \to \infty$, this procedure yields a two-dimensional Fokker–Planck equation for $P(\rho,x,t)$, governing the joint stochastic dynamics of the  settled and explorer effective population densities $\rho,\, x$.

By subsequently invoking a separation of timescales between the dynamics of explorers and settled populations, which we apply to the corresponding Langevin formulation, we obtain the following one-dimensional effective  It\^o SDE:
\begin{align}
	\dot{\rho} \!=\! \mathcal{\tilde{A}}(\rho)\! +\! \sqrt{\frac{1}{M} \tilde{\mathcal{D}}(\rho)} \sigma(t), 
	\quad \mathcal{\tilde{A}}(\rho)\! =\! \left[c\, h(f)\, (1\!-\!\rho)\! -\! e\right] \rho, 
	\quad \tilde{\mathcal{D}}(\rho)\! =\! e\! +\! \tfrac{c h(f)}{2(1+D/\lambda)} 
	\!\left[
	1\!-\!\rho^2  \!+\! 2 \tfrac{D}{\lambda} (1\!-\!\rho)
	\right]
\end{align}
The drift $\mathcal{\tilde{A}}(\rho)$ coincides with the deterministic mean-field dynamics \eqref{eqn:rate_equations_singlespecies}, while the diffusion term $\mathcal{D}(\rho)$ arises from the combination of demographic noise in both the explorer and settled compartments, properly renormalized by the fast-variable elimination.  

To derive the scaling relations presented in \cref{sec:stochastic}, we start from
the backward Fokker-Planck equation corresponding to \cref{eq:stochSDE} for the n-th moment of extinction time $T_n = \langle T^n \rangle$ reads:
\begin{align}\label{eq:nthmoment}
    &\mathcal{\tilde{A}}(\rho)\,\partial_{\rho} T_n(\rho)+\frac{1}{2 M}\tilde{\mathcal{D}}(\rho)\,\partial_{\rho}^2 T_n(\rho) = -n \, T_{n-1}(\rho)
\end{align}
and that for the survival probability:
\begin{align}\label{eq:survprob}
  \partial_t S &= \mathcal{\tilde{A}}(\rho)\,\partial_{\rho} S+\frac{1}{2 M}\tilde{\mathcal{D}}(\rho)\,\partial_{\rho}^2 S \; .
\end{align} 
Introducing the parameter
\begin{align}\label{eq:delta}
\Delta = \frac{1}{\lambda_M}\left( \lambda_M - \frac{e}{c}\right) = \frac{\epsilon}{\lambda_M} \, ,
\end{align}
to quantify the distance from the transition point, $\mathcal{\tilde{A}}$ and $\tilde{\mathcal{D}}$ can be rewritten as: 
\begin{align}
    & \mathcal{\tilde{A}}(\rho, M, \Delta) = \lambda_M \,c \left( \Delta - \rho \right) \rho \nonumber \\
    & \tilde{\mathcal{D}}(\rho, M, \Delta) = 
     \frac{\lambda_M \, c}{2}
    \,
    \left[ 
    3 + \frac{D}{D+\lambda}
    - 2 \Delta 
    - \frac{2D}{D+\lambda}  \,\rho
    - \frac{\lambda}{D+\lambda}  \rho^2 
    \right] \, \rho
\end{align}
and we assume to be close to the transition point, so $\epsilon \approx 0$ in \cref{eq:delta}.

By properly rescaling the effective population density $\rho \rightarrow \hat{\rho} = \rho\, \sqrt{M}$, the distance from the critical point $\Delta \rightarrow \hat{\Delta} = \Delta \sqrt{M}$ and the extinction-time moments in equation \eqref{eq:nthmoment} (or equivalently rescaling the survival probability in equation \eqref{eq:survprob}), we obtain the scaling exponents of \ref{tab:scaling}, which define the following scaling forms:
\begin{align}
    \begin{split}
        &  S(t|M) = t^{-\alpha} \hat{S}( t M^{\phi})
        \\
        & T_1(\rho,\Delta, M)= M^{\beta} \hat{T}_1 (\rho M^{\gamma}, \Delta M^{\eta})
        \\
        & T_2(\rho,\Delta, M)= M^{\delta} \hat{T}_2 (\rho M^{\gamma}, \Delta M^{\eta})
    \end{split} \; .
\end{align}
The scaling analysis above characterizes the critical behavior and extinction dynamics close to the transition point. 
Beyond this regime, in the survival phase ($\epsilon>0$), the system reaches a metastable state where demographic fluctuations persist around a finite mean occupancy. 
To quantify these fluctuations and connect them with the underlying nonlinear stochastic dynamics, we next consider a linearized approximation around the deterministic fixed point. 
This approach leads to an Ornstein–Uhlenbeck (OU) description of the quasi-stationary dynamics, 
\begin{equation}
	\dot \rho_{\mathrm{OU}} = -\kappa\, \rho_{\mathrm{OU}} + \sqrt{2\sigma_{\mathrm{eff}}} \, \eta(t),
	\qquad \text{with} \qquad \sigma_{\mathrm{eff}} = \sqrt{ \frac{D_{\mathrm{eff}}}{2M\kappa} },
\end{equation}
where \(\eta(t)\) is Gaussian white noise, and the effective drift \(\kappa\) and diffusion coefficient \(D_{\mathrm{eff}}\) are 
\begin{equation}
	\kappa = \left| \frac{\partial \mathcal{\tilde{A}}}{\partial \rho} \big\rvert_{\rho_*} \right| = |e - c\,h(f)|\,,
	\qquad
	D_{\mathrm{eff}} = \tilde{\mathcal{D}}(\rho_*) = \frac{1}{2} \frac{e\, [e - c\,h(f)] [e\, \lambda - 4 \,c\,h(f)(D + \lambda)]}{[c\,h(f)]^2 (D + \lambda)}.
\end{equation}

\begin{table}[t]
	\begin{tabular*}{0.3\linewidth}{@{\extracolsep{\fill}}cccccc}
		\hline\hline
		$\alpha$ 
		&$\phi$ 
		&$\gamma$ 
		&$\eta$ 
		&$\beta$
		&$\delta$
		\rule[-2ex]{0pt}{5ex} \\
		\hline
		0 & $-\frac{1}{2}$ &$\frac{1}{2}$ &$\frac{1}{2}$
		&$\frac{1}{2}$ &1
		\rule[-2ex]{0pt}{5ex}
		\rule[-2ex]{0pt}{5ex} 
		\\
		\hline\hline
	\end{tabular*}
	\caption{
		\raggedright
		Scaling exponents characterizing the survival probability $S$ and the first and second moments of the exit time, $\langle T \rangle$ and $\langle T^2 \rangle$ in the vicinity of the critical point.
	}
	\label{tab:scaling}
\end{table}

\end{document}